\documentclass[letterpaper, 11pt]{article}

\usepackage{amsmath, amssymb, amsfonts, amsthm}
\usepackage{bm}
\usepackage{mathtools}
\usepackage{mathrsfs}
\usepackage{geometry}
\usepackage{setspace}
\usepackage[utf8x]{inputenc}
\usepackage{algpseudocode,algorithmicx,algorithm}
\usepackage[T1]{fontenc}
\usepackage{lmodern}
\usepackage{graphicx}
\usepackage{enumerate}
\usepackage[natural, dvipsnames, pdftex]{xcolor}
\usepackage{tikz}
\usepackage{verbatim}
\usepackage{hyperref}
\usepackage{cancel}
\usepackage[all]{xy}
\usepackage{mdframed}
\usepackage[capitalize]{cleveref}
\usepackage{todonotes}
\usepackage{algpseudocode}
\usepackage{tikz-cd}

\usepackage[numbers]{natbib}

\usepackage{enumitem}

\usepackage{indentfirst}
\theoremstyle{plain}
\newtheorem{thm}{Theorem}[section]
\newtheorem{lemma}[thm]{Lemma}
\newtheorem{prop}[thm]{Proposition}
\newtheorem{cor}[thm]{Corollary}

\theoremstyle{definition}
\newtheorem{defn}[thm]{Definition}

\newtheorem{ex}[thm]{Example}

\crefname{thm}{theorem}{theorems}
\Crefname{thm}{Theorem}{Theorems}

\theoremstyle{remark}
\newtheorem{rmk}[thm]{Remark}

\newcommand{\N}{\mathbb{N}}

\newcommand{\F}{\mathbb{F}}

\newcommand{\cC}{\mathcal{C}}
\newcommand{\cO}{\mathcal{O}}

\newcommand{\id}{\mathrm{id}}

\newcommand{\wt}{\mathrm{wt}}
\newcommand{\rank}{\mathrm{rank}}

\newcommand{\V}{\mathcal{V}}

\def\g{{\mathfrak{g}}}

\begin{document}

\title{Thresholds for Local Coordinate-Wise Linear Properties of Random One-Point AG Codes
\thanks{
C. Yuan is with School of Computer Science, Shanghai Jiao Tong University. (Email: \href{mailto:chen_yuan@sjtu.edu.cn}{chen\_yuan@sjtu.edu.cn}) R. Zhu is with School of Computer Science, Shanghai Jiao Tong University. (Email: \href{mailto:sjtuzrq7777@sjtu.edu.cn}{sjtuzrq7777@sjtu.edu.cn})
}}

\author{Chen Yuan, Ruiqi Zhu}
\date{}
\maketitle

\begin{abstract}
Local coordinate-wise linear (LCL) properties provide a unified framework for studying list decoding, list recovery and other local properties of linear codes. We prove that random one-point algebraic geometry (AG) codes have the same LCL threshold rates as random linear codes under explicit alphabet and sampling conditions. More precisely, for a family of local profiles with random linear code threshold $R_{\mathcal P}$, a random one-point AG code avoids all such profiles below $R_{\mathcal P}-\varepsilon$ and contains one above $R_{\mathcal P}+\varepsilon$, with explicit failure bounds on both sides. Our proof also removes an auxiliary lower bound on the number of available rational places that arises in a direct adaptation of the Reed--Solomon argument and requires no additional genus-dependent rate gap.

We also observe that the previously claimed above threshold proof for random Reed--Solomon codes does not establish the required probability lower bound. The containment relation used in that argument gives the probability comparison in the opposite direction from what is needed. We give a different above threshold argument, which in particular also yields an alternative proof in the genus zero case.

Combining our threshold theorem with recent results for random linear codes gives a collection of new list decoding, average weight list decoding and list recovery guarantees for random one-point AG codes. We further obtain new correlated agreement and proximity gap results for polynomial curves. These results extend the LCL threshold framework from random linear and Reed--Solomon codes to random one-point AG codes.
\end{abstract}


\section{Introduction}
Error-correcting codes are fundamental to reliable communication and storage in the presence of noise. Among the central notions in coding theory are list decoding and list recovery, which relax unique decoding by allowing the decoder to output a small set of candidate codewords. These notions play important roles not only in coding theory~\cite{guruswami1998improved,guruswami2021bounds}, but also in complexity theory, pseudorandomness~\cite{ta2001extractor,guruswami2009unbalanced,indyk2010efficiently} and proof systems~\cite{ben2018fast,ben2023proximity,ames2017ligero}.

List decoding, first introduced in~\cite{elias1957list}, provides a way to recover codewords even when the number of errors exceeds half of the minimum distance. A code $\cC\subseteq\F_q^n$ is $(\rho,L)$-list-decodable if, for every $\mathbf y\in\F_q^n$,
\[
|\{\mathbf x\in\cC:\wt(\mathbf x-\mathbf y)\le\rho n\}|\le L.
\]
List decoding has led to a rich collection of combinatorial and algorithmic results~\cite{guruswami1998improved,guruswami2008explicit,guruswami2011linear,parvaresh2005correcting}. It also provides tools for low degree testing and proximity arguments in proof systems~\cite{ben2023proximity,gao2025list}. List recovery allows a list of possible symbols at each coordinate instead of a single received symbol. For $\ell\in\mathbb N$, the code $\cC$ is $(\rho,\ell,L)$-list recoverable if, for every $Z_1,\ldots,Z_n\subseteq\F_q$ with $|Z_i|\le\ell$,
\[
\left|\left\{\mathbf x\in\cC:
|\{i\in[n]:x_i\notin Z_i\}|\le\rho n\right\}\right|\le L.
\]
In particular, $(\rho,L)$-list decodability is equivalent to $(\rho,1,L)$-list recoverability. List recovery has been used to obtain efficient list decoding algorithms~\cite{guruswami1998improved,guruswami2008explicit,guruswami2013linear,kopparty2015list}, and has applications in complexity theory and algorithms~\cite{TaShmaZuckerman2004,guruswami2009unbalanced,GilbertEtAl2013,indyk2010efficiently}.

Random linear codes provide a benchmark for many such decoding tasks. A large body of work establishes strong list-decoding and list-recovery guarantees for these codes, with recent improvements in the leading list size term, average radius guarantees and alphabet size requirements~\cite{guruswami2021bounds,yuan2026list,guruswami2026average,doron2025list,li2025near,goyal2026counting}. A natural goal is therefore to understand whether structured algebraic code families have the same LCL threshold rates as random linear codes.

Among algebraic code families, one-point algebraic-geometry (AG) codes retain an evaluation structure while allowing lengths that grow over a fixed alphabet~\cite{tsfasman1982modular,BrakensiekDharGopiZhang2024AG}. This makes them attractive when the alphabet size must be independent of the block length. Randomly punctured AG codes are already known to achieve list-decoding capacity over constant-sized fields~\cite{BrakensiekDharGopiZhang2024AG}. Our question is whether the LCL thresholds of random linear codes continue to hold for random one-point AG codes.

Local coordinate-wise linear (LCL) properties provide a framework for this question. Building on earlier work on local properties~\cite{mosheiff2020ldpc,guruswami2022punctured}, Levi et al.~\cite{levi2024random} formulate local obstructions as matrices whose rows satisfy prescribed linear constraints. This formulation encompasses failures of list decoding, list recovery, and their average-weight variants, as well as properties used to derive proximity gaps~\cite{goyal2025optimal,yuan2026proximity}. The framework assigns a threshold rate to each family of profiles and relates random Reed--Solomon codes to the random-linear-code model.

In this work, we prove that random one-point AG codes have the same LCL threshold rates as random linear codes under explicit alphabet and sampling conditions. Below the threshold, with high probability the code avoids all profiles in the given family, while above the threshold it contains at least one such profile. To prove this in the function field setting, we adapt the algebraic arguments used for Reed--Solomon codes to Riemann--Roch spaces. For the above threshold direction, we instead compare the dimensions of the full space of matrices satisfying the profile and the subspaces in which two columns are equal. This avoids a gap in the existing random Reed--Solomon proof, which we explain below.

\subsection{Background on Random Linear Codes and One-Point AG Codes}
A random linear code of dimension $k=Rn$ can be sampled as the kernel of a uniformly random full rank matrix in $\F_q^{(n-k)\times n}$. Random linear codes achieve the optimal rate and radius tradeoff for list decoding, and a central question is how small their output list size can be~\cite{guruswami2010list,cheraghchi2013restricted,
wootters2013list,li2018improved}. 

For fixed $q$ and $\rho\in(0,1-1/q)$, the rate is $1-H_q(\rho)-\varepsilon$, where $H_q$ is the $q$-ary entropy function. Yuan and Zhu~\cite{yuan2026list} and Silas~\cite{silas2026listsize} show that random linear codes have list size $\left\lceil\frac{H_q(\rho)}{\varepsilon}\right\rceil+O_{\rho,q}(1),$ matching the leading term in the lower bound of~\cite{guruswami2021bounds}. Guruswami et al.~\cite{guruswami2026average} prove the stronger average radius guarantee with list size $O_{\rho,q}(1/\varepsilon)$ over every prime power alphabet. Kumar and Mon~\cite{kumar2026list} prove the uniform bound $O(qH_q(\rho)/\varepsilon)$ on the list size of random linear codes of rate $1-H_q(\rho)-\varepsilon$, valid for all $\rho$, $q$ and $\varepsilon$. For large alphabets, Goyal and Guruswami~\cite{goyal2026counting} improve the alphabet requirements for approaching the generalized Singleton bound.

Several arguments also show that random linear codes achieve list recovery capacity. Starting from the bounds of Zyablov and Pinsker~\cite{zyablov1981list}, Rudra and Wootters~\cite{rudra2018average} improve the output list size, including in the average radius setting. More recently, Li and Shagrithaya~\cite{li2025near} obtain output list sizes independent of the alphabet size. Doron et al.~\cite{doron2025list} establish $O(1/\varepsilon)$ output list size for fixed field size, input list size, and error radius. Brakensiek et al.~\cite{brakensiek2025combinatorial} use discrete Brascamp-Lieb inequalities to obtain near-optimal list-recovery bounds for several linear and folded-linear code families. Goyal and Guruswami~\cite{goyal2026counting} improve the alphabet size needed for random linear codes to approach the generalized Singleton bound, for both average weight list decoding and list recovery. Section~\ref{LCL-property} records the specific results used in our applications.

Progress for algebraic codes has been equally substantial. Generic Reed--Solomon codes and randomly punctured Reed--Solomon codes are known to achieve list decoding capacity~\cite{brakensiek2023generic,GuoZhang2023PuncturedRS}. More recently, Alrabiah et al.~\cite{alrabiah2024randomly} showed that, in the randomly punctured setting, a field of size linear in the block length suffices when a constant rate slack is allowed. Chen and Zhang~\cite{chen2024folded} prove relaxed generalized Singleton bounds for explicit folded Reed--Solomon and multiplicity codes. Recent work of Brakensiek et al.~\cite{brakensiek2026algorithmic}, and of Jeronimo~\cite{jeronimo2026algorithmic}, gives capacity approaching algorithms for ordinary Reed--Solomon codes over prime fields on prescribed evaluation sets.

Let $F/\F_q$ be a function field, let $P_1,\ldots,P_n$ be distinct rational places, and let $G$ be a divisor whose support avoids these places. The associated AG code is
\[
C_L(D;G):=\{(f(P_1),\ldots,f(P_n)):f\in\mathcal L(G)\},
\qquad D=P_1+\cdots+P_n.
\]
We study $G=mP_\infty$, with the evaluation places sampled uniformly without replacement from the rational places other than $P_\infty$. Guruswami-Sudan type methods and subsequent work provide list decoding algorithms for AG codes~\cite{guruswami1998improved,ShokrollahiWasserman1999,CohnHeninger2011,GeilMatsumotoRuano2012,BeelenRosenkildeSolomatov2022}. Folded AG constructions and subcode methods yield further capacity approaching results~\cite{GuruswamiXing2013,GuruswamiXing2014,GuoRonZewi2021Efficient}.

For ordinary randomly punctured AG codes, Brakensiek et al.~\cite{BrakensiekDharGopiZhang2024AG} prove list decoding capacity with list size $O(1/\varepsilon)$ and field size $\exp(O(1/\varepsilon^2))$, using relaxed higher order MDS codes and a generalized GM MDS theorem. Our aim is not to improve this alphabet bound, but to obtain a threshold theorem that applies to general LCL properties of random one-point AG codes.

\subsection{Local Coordinate-Wise Linear Properties}

Local properties of codes were introduced in~\cite{mosheiff2020ldpc} to study list decoding and list recovery in a common framework. Mosheiff et al.~\cite{mosheiff2020ldpc} showed that random Gallager LDPC codes inherit suitable local properties of random linear codes, and Guruswami and Mosheiff~\cite{guruswami2022punctured} proved an analogous result for codes obtained by randomly puncturing low-bias linear codes. Levi et al.~\cite{levi2024random} develop the more general LCL formulation: a profile $\V=(V_1,\ldots,V_n)$ prescribes a linear subspace $V_i\subseteq\F_q^b$ for each row of an $n\times b$ matrix. Containment requires a matrix satisfying the profile whose columns are distinct codewords. In this paper, $\mathcal P$ denotes this containment property, so its complement is the desired guarantee in decoding applications.

Recent work extends this perspective beyond random Reed--Solomon codes. Jeronimo and Shagrithaya~\cite{jeronimo2025local} use an LCL formulation and the Alon-Edmonds-Luby construction to obtain explicit codes over enlarged constant alphabets. Brakensiek et al.~\cite{brakensiek2025subspace} show that nearly optimal subspace design codes inherit local properties of random linear codes. Goyal et al.~\cite{goyal2026constant} construct explicit constant alphabet subspace design codes, and Jeronimo and Shagrithaya~\cite{jeronimo2026subspace} show that the subspace design property can be expressed as an LCL property, and use their earlier construction to obtain constant alphabet subspace design codes. These works demonstrate the scope of local property methods.

Proximity gaps and correlated agreement provide another motivation for local property analysis. Following~\cite{ben2023proximity}, Gao et al.~\cite{gao2025list} derive proximity gap results from list decodability, while Goyal and Guruswami~\cite{goyal2025optimal} obtain stronger proximity gap bounds using subspace design and local profile methods. Recent results also clarify limitations for general Reed--Solomon codes and the relation between proximity gaps and list decoding~\cite{ben2025on,crites2025conjectures}. Yuan and Zhu~\cite{yuan2026proximity} develop a syndrome-space approach for random linear and random Reed--Solomon codes, including polynomial curves, using LCL profiles in the random Reed--Solomon case. Jeronimo~\cite{jeronimo2026algorithmic} obtains algorithmic list decoding and curve mutual correlated agreement for Reed--Solomon codes over prime fields. In Section~\ref{subsec:AG-applications}, we adapt the curve argument of~\cite{yuan2026proximity} using our AG profile bound.

\subsection{Our Contributions}

The starting point is the random linear code threshold $R_{\mathcal P}$ associated with a $b$-local profile family $\mathcal F$ in~\cite{levi2024random}. We prove that the same threshold governs random one-point AG codes on both sides, with explicit failure bounds.

\paragraph{Informal theorem (Theorem~\ref{thm:threshold-AG-codes}).}
Fix $b\ge2$ and $\varepsilon>0$. Let $\mathcal P$ be a $b$-LCL property with a nonempty local profile family $\mathcal F$ of size at most $\exp(O(n))$ and let $R_{\mathcal P}$ be its random linear code threshold. Let $\mathcal C=\cC_{\mathcal{L}}(D;mP_\infty)\subseteq\mathbb F_q^n$ be a random one-point AG code of rate $R=\frac{\ell(mP_\infty)}{n}$ where $0\le m<n\le N$ and $N$ is the number of available evaluation places.

If $q>\binom b2$ and $N-n+1$ is a sufficiently large multiple of $m$, then for all sufficiently large $n$, with probability $1-\exp(-\Omega(n))$, we have
\[
\begin{cases}
\mathcal C\notin\mathcal P,
& R\le R_{\mathcal P}-\varepsilon,\\[2mm]
\mathcal C\in\mathcal P,
& R\ge R_{\mathcal P}+\varepsilon.
\end{cases}
\]
See Theorem~\ref{thm:threshold-AG-codes} for the precise parameter conditions and failure bounds. In particular, our proof requires no separate assumption $N>mb$ and no additional rate gap depending on the genus.

\paragraph{A gap in the above threshold RS argument.}
The above-threshold proof of Theorem~3.10 in~\cite{levi2024random} contains a gap. The containment relation used there gives the probability comparison in the opposite direction from what is needed to prove the required lower bound. This is a gap in the proof rather than a counterexample to the theorem.

We avoid this issue by working directly with the original profile, rather than passing to a projected profile as in~\cite{levi2024random}. We compare the full space of matrices satisfying the original profile with the subspaces in which a fixed pair of columns are equal, and show that each such subspace has strictly smaller dimension with high probability. Hence their union cannot cover the full space, which yields a matrix with pairwise distinct codeword columns. See Section~\ref{subsec:above-threshold} for the full proof.

\paragraph{Technical overview.}
We retain the local profile and potential framework of~\cite{levi2024random}, but several parts of the Reed--Solomon argument do not directly extend from polynomials to Riemann--Roch spaces.

\emph{Below the threshold.}
We first represent a finite field local profile by linear constraints on tuples of functions in $\mathcal L(mP_\infty)^b.$ The main issue is to compare the finite field dimensions appearing in the RLC threshold with the function field dimensions arising in the AG argument. A direct analogue of the RS specialization step would introduce the extra condition $N>mb$. Instead, Lemma~\ref{lem:lcl-to-Pinf}(2) shows that every $F$-linear subspace spanned by vectors in $\mathcal L(mP_\infty)^b$ can be associated with an $\F_q$-linear subspace of the same dimension, while preserving the comparison between the local constraints. This avoids searching for an evaluation place where the dimension is preserved, and hence removes the additional assumption $N>mb$.

The next task is to control how much dimension can be lost after evaluation. Let $\Omega$ denote the available rational places, and let $S\subseteq \mathcal L(mP_\infty)^b$ be an $\F_q$-linear subspace and suppose that its $F$-span has dimension $D$. For a rational place $P\in\Omega$, the evaluated space $S(P)$ is an $\F_q$-linear space whose dimension may be smaller than $D$. For Reed--Solomon codes, such a dimension loss is controlled by the number of roots of a nonzero polynomial. For AG codes, we instead use the zeros and poles of a determinant. Lemma~\ref{lem:specialization} shows that $\sum_{P\in\Omega}\bigl(D-\dim_{\F_q}S(P)\bigr)\le Dm.$ Thus, only a controlled amount of dimension can be lost over all available rational places, which bounds the probability that the next randomly sampled place causes a dimension drop.

The recursive argument also needs to pass to smaller subspaces and quotient spaces. Over a function field, an arbitrary quotient map may have poles at the sampled evaluation places. In that case, even if a function tuple satisfies the local constraint at a sampled point $P_i$, the corresponding constraint after taking the quotient may no longer be well defined there. Lemma~\ref{lem:quotient-PsiV} constructs quotient maps whose coefficients are regular at every sampled place and whose pole orders at $P_\infty$ are controlled. Hence, if the original tuple satisfies the local constraint at $P_i$, its image in the quotient still satisfies the corresponding induced constraint at $P_i$. This allows the same evaluation estimates to be used recursively after passing to a quotient. Together with the dimension comparison above, these results yield Theorem~\ref{thm:C_contain_Pr}, which bounds the probability that a nonzero matrix with columns in $\mathcal C$ satisfies a profile whose potential is sufficiently negative on every nonzero subspace. Section~\ref{subsec:below-threshold} then proves the below threshold part of Theorem~\ref{thm:threshold-AG-codes}.

\emph{Above the threshold.}
Choose $\V$ and $U$ with $R_{\V,U}=R_{\mathcal P}$, and write $k=\dim\cC$. Let $K_U$ consist of the matrices whose columns lie in $\cC$ and whose $i$-th row lies in $V_i\cap U$. When $R\ge R_{\mathcal P}+\varepsilon$, counting variables and constraints, together with Lemma~\ref{lem:ag-constrained-dimension}, shows that, with high probability, $K_U$ has strictly larger dimension than each of its subspaces in which a fixed pair of columns is equal. The key idea is to reduce the desired dimension bound to the absence of nonzero solutions in smaller one-point AG subcodes. To prove the dimension bound for a fixed subspace $W\subseteq U$, we choose a chain $\{0\}=W_0\subsetneq W_1\subsetneq\cdots\subsetneq W_s=W$ and treat each quotient $W_j/W_{j-1}$ separately, writing $d_j=\dim W_j-\dim W_{j-1}$.

For each quotient, we choose a smaller subcode $\mathcal C_j$ of dimension $k_j$. The ambient space $\mathcal C^{d_j}$ has dimension $kd_j$, while $\mathcal C_j^{d_j}$ has dimension $k_jd_j$. Therefore, if the space of matrices in $\mathcal C^{d_j}$ satisfying the induced profile had dimension greater than $(k-k_j)d_j$, it would necessarily contain a nonzero matrix from $\mathcal C_j^{d_j}$. Theorem~\ref{thm:C_contain_Pr} shows that this happens only with small probability. Summing the resulting dimension bounds over the quotients gives Lemma~\ref{lem:ag-constrained-dimension}.

Since there are only $\binom b2<q$ possible pairs of equal columns, these smaller subspaces cannot cover the full solution space. Hence there exists a matrix satisfying the profile whose codeword columns are pairwise distinct. This proves the above-threshold part of Theorem~\ref{thm:threshold-AG-codes}; see Lemma~\ref{lem:ag-constrained-dimension} and Section~\ref{subsec:above-threshold} for the details. Unlike the argument in~\cite{levi2024random}, we use quotient spaces only to derive dimension bounds and never transfer a probability lower bound from a projected profile back to the original one, thereby avoiding the reversed probability comparison in the above threshold proof of Theorem~3.10.

\paragraph{Consequences.}

Combining the below-threshold theorem with recent random linear code results gives new list decoding, average weight list decoding and list recovery results for random one-point AG codes. Using the argument for polynomial curves from Yuan and Zhu~\cite{yuan2026proximity}, we also obtain curve correlated agreement and curve proximity gaps. These consequences are stated in Section~\ref{subsec:AG-applications}.

\section{Preliminaries}
\subsection{General notation}\label{subsec:prelims-notation}

Let $[n]$ be the set $\{1,\ldots,n\}$ and $\F_q$ be the finite field of order $q$, where $q$ is a prime power. For a vector space $V$ over $\F_q$, we denote by $\mathcal{L}(V)$ the set of all $\F_q$-linear subspaces of $V$. For $x\in[0,1]$, define the $q$-ary entropy function by $H_q(x)\coloneqq x\log_q(q-1)-x\log_q x-(1-x)\log_q(1-x),$ with the convention $0\log_q 0=0$.

Given $b\in\mathbb{N}$, an $n\times b$ matrix $A\in\F_q^{n\times b}$, and a code $\cC\subseteq \F_q^n$, we write $A\subseteq \cC$ to mean that every column of $A$ lies in $\cC$. We use $A_{i,*}$ and $A_{*,j}$ to denote the $i$-th row and $j$-th column of $A$, respectively, and define
\[
\mathrm{Dist}_q^{n,b}\coloneqq \{A\in\F_q^{n\times b}:\ A_{*,1},\dots,A_{*,b}\text{ are pairwise distinct}\}.
\]
For $\mathbf x,\mathbf y\in\F_q^n$, let $\wt(\mathbf x-\mathbf y)$ denote their Hamming distance. For a code $\mathcal C\subseteq\F_q^n$, define $d(\mathbf x,\mathcal C):=\min_{\mathbf c\in\mathcal C}\wt(\mathbf x-\mathbf c).$

\subsection{Function fields and random one-point AG codes}

Let $F/\F_q$ be a function field of genus $\g$ and let $N(F)$ denote its number of rational places. For a divisor $G$ on $F$, define the Riemann--Roch space
\[
\mathcal{L}(G)
\coloneqq \{0\}\cup\{f\in F^\times:(f)+G\ge 0\},
\]
where $(f)$ denotes the principal divisor of $f$.

For a place $P$ of $F$, let $\cO_P$ denote the local ring of functions regular at $P$, i.e., the elements of $F$ having no pole at $P$, together with $0$. A \emph{local parameter} at $P$ is an element of $\cO_P$ having a simple zero at $P$.

\begin{defn}\label{def:random-one-point-ag}
Given pairwise distinct rational places $P_1,\dots,P_n$ outside the support of $G$, set $D=P_1+\cdots+P_n$. The associated algebraic-geometry (AG) code is
\[
\cC_{\mathcal{L}}(D;G)
\coloneqq
\{(f(P_1),\dots,f(P_n)):f\in\mathcal{L}(G)\}
\subseteq\F_q^n.
\]
When $G=mP_\infty$ for a fixed rational place $P_\infty$, this is a \emph{one-point AG code}. If $P_1,\dots,P_n$ are sampled uniformly without replacement from the rational places of $F$ other than $P_\infty$, with $n\le N(F)-1$, we call the resulting code a \emph{random one-point AG code}.
\end{defn}

We will use families of function fields with many rational places relative to their genus. We therefore recall the Drinfeld--Vl\u{a}du\c{t} bound and the Garcia--Stichtenoth tower, which attains this bound when $q$ is a square.

\begin{lemma}[Drinfeld--Vl\u{a}du\c{t} bound~\cite{DV83}]
For every prime power $q$, the Ihara constant
\[
A(q)\coloneqq \limsup_{\g\to\infty}\ \max_{F/\F_q:\ \g(F)=\g}\ \frac{N(F)}{\g}
\]
satisfies $A(q)\le \sqrt{q}-1$.
Equivalently, for any family $\{F_i/\F_q\}_{i\ge 1}$ with $\g(F_i)\to\infty$,
\[
\limsup_{i\to\infty}\frac{N(F_i)}{\g(F_i)}\le \sqrt{q}-1
\qquad\text{and hence}\qquad
\liminf_{i\to\infty}\frac{\g(F_i)}{N(F_i)}\ge \frac{1}{\sqrt{q}-1}.
\]
\end{lemma}

\begin{lemma}[Garcia--Stichtenoth tower~\cite{GS95}]
Assume that $q$ is a perfect square. Then there exists an explicit tower of function fields $\{F_i/\F_q\}_{i\ge 1}$ with $\g_i\coloneqq \g(F_i)\to\infty$ and $N_i\coloneqq N(F_i)$ such that $\lim_{i\to\infty}\frac{N_i}{\g_i}=\sqrt{q}-1$.
\end{lemma}

\subsection{Average Weight List-Decodability and List-Recoverability}


For a predicate $\mathbf{1}[\cdot]$, we write $\mathbf{1}[E]\in\{0,1\}$ for the indicator of an event $E$.

\begin{defn}
Let $X\subseteq \F_q^n$ be a nonempty finite set and let $\rho\in[0,1]$.
We say that $X$ is \emph{$\rho$ average weight clustered} if there exists $\mathbf{z}\in\F_q^n$ such that $\frac{1}{|X|}\sum_{\mathbf{x}\in X}\wt(\mathbf{x}-\mathbf{z})\le \rho n.$

Let $\ell\in\N$. We say that $X$ is \emph{$(\rho,\ell)$ average weight recovery clustered} if there exist sets $Z_1,\ldots,Z_n\subseteq \F_q$ with $|Z_i|\le \ell$ for all $i\in[n]$ such that
\[
\frac{1}{|X|}\sum_{\mathbf{x}=(\mathbf{x}_1,\ldots,\mathbf{x}_n)\in X}\ \sum_{i=1}^n \mathbf{1}[\mathbf{x}_i\notin Z_i]\ \le\ \rho n .
\]
\end{defn}

\begin{defn}
Let $\cC\subseteq \F_q^n$ be a code, and let $L\in\N$. We say that $\cC$ is $(\rho,L)$ average weight list decodable if $\cC$ contains no $\rho$ average weight clustered subset of size $L+1$. We say that $\cC$ is $(\rho,\ell,L)$ average weight list-recoverable if $\cC$ contains no $(\rho,\ell)$ average weight recovery clustered subset of size $L+1$.
\end{defn}

\begin{rmk}
If $\cC$ is $(\rho,L)$ average weight list decodable, then it is $(\rho,L)$ list decodable. Similarly, $(\rho,\ell,L)$ average weight list recoverable implies $(\rho,\ell,L)$ list recoverable.
\end{rmk}

\subsection{Local Coordinate-Wise Linear Properties}\label{def:LCL}

A (code) property of $q$-ary length $n$ codes is a subset $\mathcal{P}$ of the set of $\F_q$-linear subspaces of $\F_q^n$. A code $\cC\subseteq \F_q^n$ satisfies $\mathcal{P}$ if $\cC\in \mathcal{P}$.

Fix $b\in\mathbb N$. A $b$-local profile is a sequence $\V=(V_1,\dots,V_n)\in \mathcal{L}(\F_q^b)^n$, where $\mathcal{L}(\F_q^b)$ denotes the family of $\F_q$-linear subspaces of $\F_q^b$. A matrix $A\in \F_q^{n\times b}$ satisfies $\V$ if its $i$-th row $A_{i,*}\in V_i$ for all $i\in[n]$. Define $\mathcal{M}_{\V}:=\{A\in \F_q^{n\times b}:\ A\text{ satisfies }\V\}$ and $\mathcal{M}_{\V}^{\mathrm{Dist}}:=\mathcal{M}_{\V}\cap \mathrm{Dist}_q^{n,b}.$ We say a code $\cC\subseteq \F_q^n$ contains $\V$ if there exists $A\in \mathcal{M}_{\V}^{\mathrm{Dist}}$ such that $A\subseteq \cC$.

A property $\mathcal P$ is called a $b$-LCL property if there exists a family $\mathcal{F}\subseteq \mathcal{L}(\F_q^b)^n$ such that
\[
\mathcal{P}=\{\, \cC\subseteq \F_q^n:\ \exists\,\V\in\mathcal{F} \text{ with } \cC \text{ containing } \V \,\}.
\]
In other words, the property $\mathcal{P}$ is $b$-LCL if it consists of those codes that satisfy at least one profile from a certain family of $b$-local profiles.

\begin{prop}[Proposition~2.2~\cite{levi2024random}]\label{prop:LCL-decoding-recovery}
The following holds:
\begin{enumerate}
    \item The complement of $(\rho,L)$-list-decodability is an
    $(L+1)$-LCL property with an associated $(L+1)$-local profile
    family of size at most $\binom{n}{\lfloor\rho n\rfloor}^{L+1}$.

    \item The complement of $(\rho,\ell,L)$-list-recoverability is an
    $(L+1)$-LCL property with an associated $(L+1)$-local profile
    family of size at most
    $\binom{n}{\lfloor\rho n\rfloor}^{L+1}\cdot \ell^{(L+1)n}$.
\end{enumerate}
\end{prop}

\section{LCL Properties of Random Linear Codes}\label{LCL-property}
We recall threshold results for LCL properties of random linear codes from~\cite{levi2024random}. These results will be used to study the corresponding thresholds for random one-point AG codes. In our applications to list decoding and list recovery, $\mathcal P$ denotes the failure of the corresponding decoding guarantee. Thus, an upper bound on $\Pr[\cC\in\mathcal P]$ gives an upper bound on the probability that the code fails to satisfy the desired decoding guarantee.

For the quoted RLC bounds, we use the sampling convention of the corresponding cited result. In particular, Proposition~3.1 and Theorem~3.2 follow~\cite{levi2024random}, where $\mathcal C=\ker H$ with $H\in\F_q^{(n-k)\times n}$ is chosen uniformly at random. In this model, $\dim\mathcal C\ge k$ and may be larger than $k$ if $H$ is not full row rank. Another common model chooses $\mathcal C$ uniformly from all $k$-dimensional subspaces of $\F_q^n$. This is exactly the first model conditioned on $H$ having full row rank. Since $\Pr[\rank H<n-k]<\frac{q^{-k}}{q-1},$ the two models give the same high-probability conclusions for fixed positive rate as $n\to\infty$.

Fix $b\in\N$ and a $b$-local profile $\V$. We group the matrices in $\mathcal M_{\V}$ according to their row spaces. For each $U\in\mathcal L(\F_q^b)$, define $\mathcal M_{\V,U}\coloneqq\{A\in\mathcal M_{\V}:\operatorname{rowspan}(A)=U\}.$ We say that a code $\cC\subseteq\F_q^n$ contains $(\V,U)$ if there exists $A\in\mathcal M_{\V,U}$ such that $A\subseteq\cC$. Moreover, we define
\[
\mathcal{L}_{\mathrm{Dist}}(\F_q^b) := \Bigl\{U\in\mathcal{L}(\F_q^b) :  \forall\,1\le i<j\le b,\ \exists\,x\in U\text{ with }\boldsymbol{x}_i\ne \boldsymbol{x}_j\Bigr\}.
\]
Then $\cC$ contains a profile $\V$ iff it contains $(\V,U)$ for some $U\in\mathcal{L}_{\mathrm{Dist}}(\F_q^b)$. Therefore,
\[
\max_{U\in\mathcal L_{\mathrm{Dist}}(\F_q^b)} \Pr[\cC\text{ contains }(\V,U)]\;\le\; \Pr[\cC\text{ contains }\V]\;\le\; \sum_{U\in\mathcal{L}_{\mathrm{Dist}}(\F_q^b)} \Pr[\cC\text{ contains }(\V,U)].
\]

\begin{prop}[Proposition~4.3~\cite{levi2024random}]\label{prop:M}
Let $n\in\mathbb{N}$ and $q$ be a prime power. Let $\cC\subseteq \F_q^n$ be a RLC of rate $R\in[0,1]$.
Fix $b\in\mathbb{N}$, a nonzero subspace $U\in\mathcal{L}(\F_q^b)\setminus\{\{0\}\}$, and a $b$-local profile $\V=(V_1,\dots,V_n)\in \mathcal{L}(\F_q^b)^n$.
Define the potential function $\Phi(\V,U,R):= \sum_{i=1}^n \dim(V_i\cap U) - (1-R)n\cdot \dim U,$ and $M(\V,U,R) := \min_{\substack{W\in\mathcal{L}(\F_q^b)\\ W\subsetneq U}} \bigl(\Phi(\V,U,R)-\Phi(\V,W,R)\bigr).$ Then:
\begin{enumerate}
\item If $M(\V,U,R)<0$, then $\Pr\!\bigl[\cC\text{ contains }(\V,U)\bigr]\le q^{\,M(\V,U,R)}$.
\item If $M(\V,U,R)>0$, then $\Pr\!\bigl[\cC\text{ contains }(\V,U)\bigr]\ge 1-q^{-M(\V,U,R)+b^2}$.
\end{enumerate}
In particular, the threshold rate for containing $(\V,U)$ is
\[
R_{\mathcal V,U} := \max_{\substack{W\in\mathcal{L}(\F_q^b)\\ W\subsetneq U}}\left\{1-\frac{\sum_{i=1}^n\bigl(\dim(V_i\cap U)-\dim(V_i\cap W)\bigr)}{n\cdot(\dim U-\dim W)}\right\},
\]
in the sense that $M(\mathcal V,U,R)>0$ if and only if $R>R_{\mathcal V,U}$.
\end{prop}

Since $\V$ is contained iff $(\V,U)$ is contained for some $U\in\mathcal{L}_{\mathrm{Dist}}(\F_q^b)$, we define the threshold rate for $\V$ by
\[
\begin{aligned}
R_{\mathcal{V}}
&=\min_{U\in \mathcal{L}_{\mathrm{Dist}}(\F_q^b)} R_{\mathcal{V},U} \notag\\
&=\min_{U\in \mathcal{L}_{\mathrm{Dist}}(\F_q^b)}\max_{\substack{W\in \mathcal{L}(\F_q^b)\\ W\subsetneq U}}\left\{1-\frac{\sum_{i=1}^n \bigl(\dim(V_i\cap U)-\dim(V_i\cap W)\bigr)}{n\cdot(\dim U-\dim W)}\right\}.
\end{aligned}
\]

Finally, let $\mathcal{P}$ be a $b$-LCL property specified by a nonempty profile family $\mathcal{F}\subseteq\mathcal{L}(\F_q^b)^n$, i.e., $\cC$ satisfies $\mathcal{P}$ iff it contains some $\V\in\mathcal{F}$.
We define the corresponding threshold rate
\begin{equation}\label{eq:R_P}
R_{\mathcal{P}}=\min_{\mathcal{V}\in \mathcal{F}} R_{\mathcal{V}}=\min_{\mathcal{V}\in \mathcal{F}}\min_{U\in \mathcal{L}_{\mathrm{Dist}}(\F_q^b)} R_{\mathcal{V},U}.
\end{equation}

We will use $R_{\mathcal P}$ in \eqref{eq:R_P} as the threshold parameter for $b$-LCL properties of random linear codes. We next recall the corresponding threshold theorem from \cite{levi2024random} and \cite{li2025near}, and then instantiate it for the $b$-LCL formulations of list-decodability and list-recoverability.

\begin{thm}[Theorem~4.4 \cite{levi2024random}]\label{thm:RLC-threshold}
Let $\mathcal{P}$ be a $b$-LCL property of codes in $\F_q^n$ and let $\mathcal{F} \subseteq \mathcal{L}(\F_q^b)^n$ be the corresponding nonempty family of profiles. Let $\cC \subseteq \F_q^n$ be an RLC of rate $R$.
Define $R_{\mathcal{P}}$ as in \eqref{eq:R_P}. The following now holds:
\begin{enumerate}
\item If $R \ge R_{\mathcal{P}}+\varepsilon$, then $\Pr\bigl[ \cC \text{ satisfies } \mathcal{P} \bigr] \ge 1 - q^{-\varepsilon n + b^2}.$

\item If $R \le R_{\mathcal{P}}-\varepsilon$, then $\Pr\bigl[ \cC \text{ satisfies } \mathcal{P} \bigr] \le |\mathcal{F}| \cdot q^{-\varepsilon n + b^2}.$

\item In particular, if $R \le R_{\mathcal{P}}-\varepsilon$ and $q \ge 2^{\frac{2\log_2 |\mathcal{F}|}{\varepsilon n}},$ then $\Pr\bigl[ \cC \text{ satisfies } \mathcal{P} \bigr] \le q^{-\varepsilon n/2 + b^2}.$
\end{enumerate}
\end{thm}

\begin{thm}[Theorem~3.4~\cite{levi2024random}]
Fix $\rho \in [0,1]$ and $L\in\N$. Consider the property $\mathcal{P}= \{\, \cC\subseteq \F_q^n \mid \cC \text{ is not } (\rho,L)\text{ list decodable}\,\}$ and the property $\mathcal{P}'= \{\, \cC\subseteq \F_q^n \mid \cC \text{ is not } (\rho,L)\text{ average weight list decodable}\,\}$.
Then,
\[
R_{\mathcal{P}} \ge R_{\mathcal{P}'} \ge \max\left\{\, 1-\rho\!\left(1+\frac{1}{L}\right),\, 0 \right\}.
\]
Furthermore, if $(1-\rho)(L+1)$ is an integer and $n$ is divisible by $\binom{L+1}{(1-\rho)(L+1)}$, then the above bound is tight, namely,
\[
R_{\mathcal{P}} = R_{\mathcal{P}'} = \max\!\left\{\, 1-\rho\!\left(1+\frac{1}{L}\right),\, 0 \right\}.
\]
\end{thm}

\begin{thm}[Theorem~1.2~\cite{yuan2026list}]
\label{thm:RLC-list-decodability-constant-q}
Fix a prime power $q$ and $p\in(0,1-1/q)$. There exists a constant $C_{p,q}>0$ such that, for every fixed $\varepsilon>0$ with $R=1-H_q(p)-\varepsilon>0$ and all sufficiently large $n$, a uniformly random linear code $\cC\subseteq\F_q^n$ of rate $R$ satisfies
\[
\Pr\left[
\cC \text{ is }
\left(p,\left\lceil\frac{H_q(p)}{\varepsilon}\right\rceil
+C_{p,q}\right)
\text{ list decodable}
\right]
\ge 1-q^{-\Omega_{p,q}(\varepsilon n)}.
\]
\end{thm}

\begin{thm}[Corollary~3.11~\cite{goyal2026counting}]
\label{thm:RLC-list-decodability-GSB}
Fix $R,\varepsilon>0$ with $R+\varepsilon<1$ and $L\in\N$.
Let $q$ be a prime power satisfying
\[
q>\max\left\{(32(L+1))^{2/\varepsilon},\frac{2(L+1)}{1-R-\varepsilon}\right\}.
\]
If $n>4L/\varepsilon$ and $Rn,\varepsilon n\in\N$, then a random linear code $\cC\subseteq\F_q^n$ of rate $R$ is $\left(\frac{L}{L+1}(1-R-\varepsilon),L\right)
\text{ average weight list decodable}$ with probability at least $1-8q^{-\varepsilon n/4}$.
\end{thm}

\begin{thm}[Theorem~1~\cite{guruswami2026average}]
\label{thm:RLC-average-list-decodability-constant-q}
Fix a prime power $q$ and $p\in(0,1-1/q)$. There exists a constant $C_{p,q}>0$ such that, for all sufficiently small $\varepsilon>0$ and sufficiently large $n$, a random linear code $\cC\subseteq\F_q^n$ of rate $1-H_q(p)-\varepsilon$ is $\left(p,\left\lceil\frac{C_{p,q}}{\varepsilon}\right\rceil\right)
\text{ average weight list decodable}$ with probability at least $1-q^{-5n}$.
\end{thm}

Write $h_{q,\ell}(x)=x\log_q\frac{q-\ell}{x}+(1-x)\log_q\frac{\ell}{1-x}.$

\begin{thm}[Theorem~31~\cite{doron2025list}]
\label{thm:RLC-list-recoverary}
Fix a prime power $q$, an integer $1\le\ell<q$, and
$\rho\in(0,1-\ell/q)$.
There exists a constant $C_{\rho,\ell,q}>0$ such that,
for every $0<\varepsilon<1-h_{q,\ell}(\rho)$ and sufficiently
large $n$, a random linear code $\cC\subseteq\F_q^n$ of rate
$1-h_{q,\ell}(\rho)-\varepsilon$ is $\left(\rho,\ell,
\left\lceil\frac{C_{\rho,\ell,q}}{\varepsilon}\right\rceil\right)
\text{-list-recoverable}$ with probability at least $1-q^{-n}$.
\end{thm}

\begin{thm}[A consequence of Theorems~4.2 and~4.4~\cite{goyal2026counting}]
\label{thm:RLC-list-recoverability-GG}
Fix $R,\varepsilon>0$ with $R+\varepsilon<1$ and $\ell\in\N$.
Set $L=\left\lceil\left(\frac{\ell}{R+\varepsilon}\right)^{2(R+\varepsilon)/\varepsilon}\right\rceil.$ Let $q>(2L)^{8/\varepsilon}$ be a prime power.
If $n\ge 8(L^2+L)/\varepsilon$ and $Rn\in\N$, then a random
linear code $\cC\subseteq\F_q^n$ of rate $R$ is
$(1-R-\varepsilon,\ell,L)$-list-recoverable
with probability at least $1-q^{-\varepsilon n/8}$.
\end{thm}

\begin{thm}[Theorem~8~\cite{li2025near}]\label{thm:RLC-list-recoverability-LS}
Fix $0<R<1$, $\varepsilon>0$ so that $1-R-\varepsilon>0$, $\ell\in\N$, and let 
$q \ge \max\left(\ell^{\frac{8R}{\varepsilon}+6},\ \ell\cdot 2^{4/\varepsilon}\right)$ be a prime power. Then, with probability at least
$1-2q^{-\varepsilon n/8}$, a random linear code $\cC\subseteq \F_q^n$ of rate $R$ is $(1-R-\varepsilon,\ell,L)$-list-recoverable with $L$ satisfying
$
L \le \left(\frac{2\ell}{\varepsilon}\right)^{2\ell/\varepsilon}.
$
\end{thm}

These two results provide complementary tradeoffs between the alphabet size and the output list size: the former gives a smaller output list, while the latter applies over smaller alphabets.

\section{Random AG Codes Near the Random-Linear Threshold}
This section expresses local profile constraints as linear conditions on functions in $\mathcal{L}(mP_\infty)$ and establishes the dimension estimates used in Section~\ref{sec:threshold-guarantees}.

We use the random one-point AG code model of Definition~\ref{def:random-one-point-ag}. Let $\Omega$ denote the set of rational places of $F$ other than $P_\infty$, and write $N=|\Omega|\ge n$. When $m<n$, evaluation at $P_1,\ldots,P_n$ is injective, so the dimension and rate of the code are $k=\ell(mP_\infty)=\dim_{\F_q}\mathcal{L}(mP_\infty)$ and $R=\frac{k}{n}.$

For $a\in\N$, we call $f=(f_1,\dots,f_a)\in F^a$ \emph{$m$-bounded at $P_\infty$} if $f_i\in \mathcal{L}(mP_\infty)$ for all $i\in[a]$, and let $Q_{m,a}= \mathcal{L}(mP_\infty)^a \subseteq F^a$ be an $\F_q$-linear space. For any rational place $P\neq P_\infty$ and any $f\in Q_{m,a}$, evaluation is well-defined and we write $f(P)= (f_1(P),\dots,f_a(P))\in \F_q^a.$ Moreover, for an $\F_q$-linear subspace $S\subseteq Q_{m,a}$ we define $S(P)= \{f(P): f\in S\}\subseteq \F_q^a.$ If $A\in F^{s\times a}$ is a matrix of functions regular at $P$, we write $A(P)\in \F_q^{s\times a}$ for the entrywise evaluation. Matrices of maps between powers of $F$ are written in the standard bases.

\begin{defn}\label{def:Pinf-bounded-map}
An $F$-linear map $\psi:F^b\to F^a$ is \emph{$P_\infty$-bounded} if it can be represented by a matrix $M_\psi\in F^{a\times b}$ whose entries have no poles away from $P_\infty$. Equivalently, there exists $t\ge 0$ such that every entry of $M_\psi$ lies in $\mathcal{L}(tP_\infty)$.
We define the $P_\infty$-degree of $\psi$ by
\[
\deg_{P_\infty}(\psi):= \min\{t\ge 0:\ \text{all entries of }M_\psi\text{ lie in }\mathcal{L}(tP_\infty)\}.
\]
In particular, if $\deg_{P_\infty}(\psi)\le t$, then for every $m\ge 0$, we have $\psi(Q_{m,b})\subseteq \mathcal{L}((m+t)P_\infty)^a.$
\end{defn}

\begin{defn}\label{def:local-Pinf-profile}
A \emph{$b$-local $P_\infty$-bounded profile} is a sequence $\Psi=(\psi_1,\dots,\psi_n)$ where each $\psi_i:F^b\to F^{b_i}$ is an $F$-linear $P_\infty$-bounded map for some integer $b_i\ge0$, and we define $\deg_{P_\infty}(\Psi):= \max_{i\in[n]}\deg_{P_\infty}(\psi_i).$

Let $\mathbf{P}=(P_1,\dots,P_n)$ be a tuple of rational places in $\Omega$. For $f\in Q_{m,b}$ we say that $f$ \emph{satisfies} $(\Psi,\mathbf{P})$ if $(\psi_i(f))(P_i)=0$ for all $i\in [n]$. For an $\F_q$-linear subspace $S\subseteq Q_{m,b}$ we define the constraint $\F_q$-linear space $S[\Psi,\mathbf{P}]:= \{f\in S:\ f \text{ satisfies }(\Psi,\mathbf{P})\}.$
\end{defn}

\begin{defn}\label{def:live-dead}
For a set $T\subseteq F^b$ and an $F$-linear subspace $U\subseteq F^b$, we say that $U$ is \emph{$T$-live} if $\mathrm{span}_F(T\cap U)=U$, and \emph{$T$-dead} otherwise. Thus, $U$ is $T$-live precisely when it has an $F$-basis consisting of elements of $T$. We take $T=Q_{m,b}$ unless stated otherwise.
\end{defn}

\begin{lemma}\label{lem:specialization}
Let $S\subseteq Q_{m,b}=\mathcal{L}(mP_\infty)^b$ be a finite dimensional $\F_q$-linear subspace and let $D= \mathrm{dim}_F(\mathrm{span}_F(S))\le b$. Then the following hold.

\begin{enumerate}
\item For every rational place $P\neq P_\infty$, $\mathrm{dim}_{\F_q} S(P)\le D$.

\item Let $\Omega$ be the set of rational places of $F$ distinct from $P_\infty$ and $|\Omega|=N$. Then $\sum_{P\in\Omega}\bigl(D-\dim_{\F_q}S(P)\bigr)\le Dm.$ In particular, if $P$ is uniform over $\Omega$, then $\Pr[\dim_{\F_q}S(P)=D]\ge 1-\frac{Dm}{N}.$

\end{enumerate}
\end{lemma}

\begin{proof}
If $S=\{0\}$ the claim is trivial. Otherwise fix an $\F_q$-basis $s^{(1)},\dots,s^{(s)}\in \mathcal{L}(mP_\infty)^b$ of $S$ and let $A\in F^{s\times b}$ be the matrix whose $j$-th row is $s^{(j)}$. For any rational $P\neq P_\infty$, all entries of $A$ are regular at $P$, so $A(P)$ is well defined and we have $t:=\mathrm{dim}_{\F_q} S(P)=\mathrm{rank}_{\F_q}(A(P)).$ Then $A(P)$ contains a $t\times t$ minor $M(P)$ with $\det(M(P))\neq 0$. Let $M$ be the corresponding $t\times t$ minor of $A$ over $F$. Since evaluation is a ring homomorphism for functions that are regular at this point, $\det(M)(P)=\det(M(P))\neq 0$. This implies $\det(M)\in F$ is nonzero and $\mathrm{rank}_F(M)=t$. Therefore, $\mathrm{dim}_{F}(\mathrm{span}_F(S))=\mathrm{rank}_F(A)\ge t$.

On the other hand, since $\mathrm{rank}_F(A)=D$, there exists a $D\times D$ minor $M'$ of $A$ with $\det(M')\neq 0$ in $F$. As each entry of $M'$ lies in $\mathcal{L}(mP_\infty)$, each term in the determinant expansion has pole order at $P_\infty$ at most $Dm$. In particular, $\det(M')$ has no poles away from $P_\infty$ and $\deg(\det(M')_0)=\deg(\det(M')_\infty)\le Dm$.

For each $P\in\Omega$, set $e_P:=D-\dim_{\F_q}S(P).$ Then
\[
\rank_{\F_q}M'(P)
\le
\rank_{\F_q}A(P)
=
\dim_{\F_q}S(P)
=
D-e_P.
\]
Thus, after invertible row operations over the local ring at $P$, we may arrange that the last $e_P$ rows of $M'$ vanish after evaluation at $P$. Hence every entry in these rows is divisible by a local parameter at $P$. Therefore the determinant is divisible by the $e_P$-th power of this local parameter. In other words, $\det(M')$ has a zero of multiplicity at least $e_P$ at $P$.

Summing these multiplicities over all $P\in\Omega$ gives
\[
\sum_{P\in\Omega}
\bigl(D-\dim_{\F_q}S(P)\bigr)
=
\sum_{P\in\Omega}e_P
\le
\deg(\det(M')_0)
\le Dm.
\]
In particular, every $P$ for which $\dim_{\F_q}S(P)<D$ contributes at least one to the sum. Hence there are at most $Dm$ such places, and therefore $\Pr_{P\in\Omega}[\dim_{\F_q}S(P)=D]\ge 1-\frac{Dm}{N}.$
\end{proof}

\begin{cor}\label{cor:simultaneous-full-dim}
Let $S\subseteq \mathcal{L}(mP_\infty)^b$ and let $D=\mathrm{dim}_F(\mathrm{span}_F(S))$ as in Lemma~\ref{lem:specialization}. Sample $P_1,\dots,P_n$ uniformly from $\Omega$ conditioned on being pairwise distinct. Then
\[
\Pr\big[\mathrm{dim}_{\F_q} S(P_i)=D\ \text{for all}\ i\in[n]\big]\ \ge\ 1-\frac{nDm}{N-n+1}.
\]
\end{cor}

\begin{proof}
Let $M'$ be the $D\times D$ minor chosen in the proof of Lemma~\ref{lem:specialization}. Its determinant is a nonzero element of $\mathcal L(DmP_\infty)$ and therefore vanishes at at most $Dm$ rational places in $\Omega$. Conditioned on any realization of $P_1,\dots,P_{i-1}$, the next place $P_i$ is uniform over $\Omega \setminus \{P_1,\ldots,P_{i-1}\}$, hence
\[
\Pr[\det(M')(P_i)=0\mid P_1,\dots,P_{i-1}] \le \frac{Dm}{N-i+1}\le \frac{Dm}{N-n+1}.
\]
A union bound over $i\in[n]$ gives
\[
\Pr\bigl[\exists i\in[n]:\det(M')(P_i)=0\bigr]
\le \frac{nDm}{N-n+1}.
\]
On the complement, $\det (M')(P_i)\neq 0$ for all $i$ implies that $\mathrm{dim} S(P_i)=D$ for all $i$ by Lemma~\ref{lem:specialization}.
\end{proof}

\begin{defn}\label{def:lift-F}
Let $\V=(V_1,\dots,V_n)\in (\mathcal{L}(\F_q^b))^n$ be a $b$-local (linear) profile. For each $i\in[n]$, let $\varphi_i:\F_q^b\to \F_q^{\,b-\dim V_i}$ be an $\F_q$-linear map with $\mathrm{ker} \varphi_i = V_i$. Since $F/\F_q$ is a field extension, scalar extension gives canonical identifications $F\otimes_{\F_q}\F_q^b \cong F^b$. Define the $F$-linear map $\psi_i := \id_F \otimes \varphi_i : F^b \to F^{b-\dim V_i}$ and let $\Psi_\mathcal{V} := (\psi_1,\dots,\psi_n)$.
\end{defn}

\begin{lemma}\label{lem:tensor-kernel}
For each $i\in[n]$, we have $\ker(\psi_i)=F\otimes_{\F_q}V_i \subseteq F\otimes_{\F_q}\F_q^b \cong F^b.$ Moreover, in the standard bases, the matrix of $\psi_i$ has entries in $\F_q$. In particular, $\psi_i$ is $P_\infty$-bounded of degree $0$.
\end{lemma}

\begin{proof}
Consider the short exact sequence of $\F_q$-vector spaces
\[
0 \longrightarrow V_i \longrightarrow \F_q^b
\xrightarrow{\;\varphi_i\;} \F_q^{b-\dim V_i}\longrightarrow 0.
\]
Since $F$ is flat over $\F_q$, tensoring with $F$ preserves exactness and yields
\[
0 \longrightarrow F\otimes_{\F_q}V_i
\longrightarrow F\otimes_{\F_q}\F_q^b
\xrightarrow{\;\id_F\otimes \varphi_i\;}
F\otimes_{\F_q}\F_q^{b-\dim V_i}\longrightarrow 0.
\]
Under the standard identifications $F\otimes_{\F_q}\F_q^b \cong F^b$ and $F\otimes_{\F_q}\F_q^{b-\dim V_i}\cong F^{b-\dim V_i},$ this is precisely the sequence defining $\psi_i$, so $\ker(\psi_i)=F\otimes_{\F_q}V_i.$

Since the matrix of $\varphi_i$ has entries in $\F_q$, the same is true for $\psi_i$. Thus, every matrix entry is constant, hence regular away from $P_\infty$, and therefore $\deg_{P_\infty}(\psi_i)=0$.
\end{proof}

\begin{lemma}\label{lem:lcl-to-Pinf}
Let $\V=(V_1,\dots,V_n)\in \bigl(\mathcal{L}(\F_q^b)\bigr)^n$ be a $b$-local linear profile and let $\Psi_\V=(\psi_1,\dots,\psi_n)$ be the associated $b$-local $P_\infty$-bounded profile obtained by scalar extension, as in Definition~\ref{def:lift-F} and Lemma~\ref{lem:tensor-kernel}. Then $\deg_{P_\infty}(\Psi_\V)=0$, and the following hold.

\begin{enumerate}
\item
Let $m<n$, let $\mathbf{P}=(P_1,\ldots,P_n)$ consist of distinct rational places of $F$ other than $P_\infty$ and set $\cC=\mathrm{ev}_{\mathbf{P}}(\mathcal{L}(mP_\infty))$. Evaluation induces an $\F_q$-linear isomorphism
\[
\mathrm{ev}_{\mathbf{P}}\colon
Q_{m,b}[\Psi_\V,\mathbf{P}]
\xrightarrow{\sim}
\cC^b\cap\mathcal{M}_\V,
\qquad
(f_1,\ldots,f_b)
\longmapsto
\bigl(f_j(P_i)\bigr)_{i\in[n],\,j\in[b]}.
\] 
In particular, $\cC$ contains $\V$ if and only if $Q_{m,b}[\Psi_\V,\mathbf{P}]$ contains a tuple of pairwise distinct functions.

\item Let $U\subseteq F^b$ be a $Q_{m,b}$-live $F$-linear subspace and let $D=\dim_F(U)$. Then there exists an $\F_q$-linear subspace $U'\subseteq\F_q^b$ with $\dim_{\F_q}(U')=D$ such that
\[
\sum_{i=1}^n\dim_F(\psi_i(U))
\ge
\sum_{i=1}^n
\bigl(D-\dim_{\F_q}(U'\cap V_i)\bigr).
\]
\end{enumerate}
\end{lemma}

\begin{proof}
Since each $\psi_i$ has coefficients in $\F_q$, we have $\deg_{P_\infty}(\Psi_\V)=0$. Since $m<n$, evaluation gives an $\F_q$-linear isomorphism $Q_{m,b}\to\cC^b$. For $f=(f_1,\ldots,f_b)\in Q_{m,b}$ and every $i\in[n]$, we have $(\psi_i(f))(P_i)=\varphi_i(f(P_i))$. Together with $\ker(\varphi_i)=V_i$, this gives
\[
(\psi_i(f))(P_i)=0
\quad\Longleftrightarrow\quad
f(P_i)\in V_i.
\]
The vector $f(P_i)$ is the $i$-th row of the evaluation matrix. Thus $f\in Q_{m,b}[\Psi_\V,\mathbf{P}]$ if and only if $\mathrm{ev}_{\mathbf{P}}(f)\in\mathcal{M}_\V$, which proves the asserted isomorphism.

The $j$-th column of the evaluation matrix is the evaluation vector of $f_j$. By injectivity of evaluation, these columns are pairwise distinct if and only if the functions $f_1,\ldots,f_b$ are pairwise distinct. This proves the claim about containment of $\V$.

For the second claim, let $U$ be $Q_{m,b}$-live and let $D=\dim_F(U)$. Let $\mathcal O_\infty$ be the local ring at $P_\infty$, and set $\Lambda:=U\cap\mathcal O_\infty^b.$ Every vector in $U$ can be multiplied by a suitable power of a local parameter at $P_\infty$ so that all of its coordinates are regular at $P_\infty$. Hence $\mathrm{span}_F(\Lambda)=U.$ Since $\mathcal O_\infty$ is a discrete valuation ring, $\Lambda$ is a free $\mathcal O_\infty$-module of rank $D$.

Moreover, reduction at $P_\infty$ preserves this rank. Indeed, if $s\in\Lambda$ vanishes at $P_\infty$, then $s=\pi s'$ for a local parameter $\pi$ and some $s'\in\mathcal O_\infty^b$. Since $U$ is $F$-linear, $s'\in U$ and hence $s'\in\Lambda$. Thus the kernel of reduction is exactly $\pi\Lambda$, and therefore $\dim_{\F_q}\Lambda(P_\infty)=D.$ Set $U':=\Lambda(P_\infty)\subseteq\F_q^b.$

Fix $i\in[n]$. Since $\ker(\varphi_i)=V_i$, rank nullity over $\F_q$ gives $\dim_{\F_q}(U'\cap V_i)=D-\dim_{\F_q}(\varphi_i(U')).$ By the definition of $U'$ and the fact that $\psi_i$ has constant coefficients, $\varphi_i(U')=\psi_i(\Lambda)(P_\infty).$

Therefore, $\dim_{\F_q}(U'\cap V_i)=D-\dim_{\F_q}\bigl(\psi_i(\Lambda)(P_\infty)\bigr).$ Since evaluation of a matrix whose entries are regular at $P_\infty$ cannot increase its rank, $\dim_{\F_q}\bigl(\psi_i(\Lambda)(P_\infty)\bigr)\le\dim_F\mathrm{span}_F\bigl(\psi_i(\Lambda)\bigr).$ Since $\mathrm{span}_F(\Lambda)=U$ and $\psi_i$ is $F$-linear, $\mathrm{span}_F\bigl(\psi_i(\Lambda)\bigr)=\psi_i(U).$ Hence, $\dim_F(\psi_i(U))\ge D-\dim_{\F_q}(U'\cap V_i).$ Summing over $i\in[n]$ gives $\sum_{i=1}^n\dim_F(\psi_i(U))\ge\sum_{i=1}^n\bigl(D-\dim_{\F_q}(U'\cap V_i)\bigr).$
\end{proof}

\begin{rmk}\label{rem:Phi-via-strength}
Let $U$ and $U'$ be as in Lemma~\ref{lem:lcl-to-Pinf}, and write $D=\mathrm{dim}_F(U)=\mathrm{dim}_{\F_q}(U')$. For any $R\in[0,1]$, the potential function on the finite field space $U'\subseteq \F_q^b$ satisfies
\[
\begin{aligned}
\Phi(\V,U',R)
&=\sum_{i=1}^n \mathrm{dim}(V_i\cap U')-(1-R)nD
=RnD-\sum_{i=1}^n\big(D-\mathrm{dim}(U'\cap V_i)\big)\\
&\ge RnD-\sum_{i=1}^n\dim_F(\psi_i(U)).
\end{aligned}
\]
\end{rmk}

\section{Rate thresholds for random one-point AG codes}\label{sec:threshold-guarantees}
We now show that a random one-point AG code satisfies an LCL property with low probability when its rate is sufficiently below the corresponding threshold for random linear codes, and with high probability when its rate is sufficiently above this threshold.

\begin{thm}\label{thm:threshold-AG-codes}
Let $b\ge2$ and let $q>\binom{b}{2}$ be a prime power. Let $\mathcal P$ be a $b$-LCL property of codes in $\F_q^n$, with nonempty profile family $\mathcal F$ and RLC threshold $R_{\mathcal P}$ as in~\eqref{eq:R_P}.

Let $F/\F_q$ have genus $\mathfrak g$ and let $\cC$ be a random one-point AG code of rate $R$ as in Definition~\ref{def:random-one-point-ag} with $0\le m<n\le N$. For every $\varepsilon>0$, the following hold.
\begin{enumerate}
\item \textbf{Below the threshold.} If $R\le R_{\mathcal P}-\varepsilon$ and $\varepsilon n\ge2b(b+1)$, then
\[
\Pr[\cC\in\mathcal P]
\le |\mathcal F|\,(2^b-1)
\left(\frac{(4b)^{4b}m}{\varepsilon(N-n+1)}\right)^{\varepsilon n/(2b)}.
\]
\item \textbf{Above the threshold.}
If $R\ge R_{\mathcal P}+\varepsilon$ and $\varepsilon n\ge 4b^2(b+1)$, then
\[
\Pr[\mathcal C\notin\mathcal P]
\le
\binom{b}{2}\,b(2^b-1)
\left(
\frac{2b(4b)^{4b}m}
     {\varepsilon(N-n+1)}
\right)^{\varepsilon n/(4b^2)}.
\]
\end{enumerate}
\end{thm}


We prove part~1 in Section~\ref{subsec:below-threshold} and part~2 in Section~\ref{subsec:above-threshold}. The list-decoding and list-recovery consequences follow in Section~\ref{subsec:AG-applications}.

\subsection{Below the random linear codes threshold}\label{subsec:below-threshold}

Fix an integer $m\ge0$ and $b\in\N$. Recall $Q_{m,b}=\mathcal{L}(mP_\infty)^b\subseteq F^b$. Let $\Psi=(\psi_1,\dots,\psi_n)$ be a $b$-local $P_\infty$-bounded profile and let $h= \deg_{P_\infty}(\Psi)$. Let $\Omega$ be any set of rational places of $F$ distinct from $P_\infty$ and $N=|\Omega|\ge n$. We sample $\mathbf{P}=(P_1,\dots,P_n)$ uniformly from $\Omega$ conditioned on being pairwise distinct. For an $\F_q$-linear subspace $S\subseteq Q_{m,b}$, recall that
\[
S[\Psi,\mathbf{P}]=\{f\in S:\ (\psi_i(f))(P_i)=0\ \text{for all }i\in[n]\}.
\]
Our goal is to derive an upper bound on $\Pr\bigl[\dim_{\F_q}(S[\Psi,\mathbf{P}])>0\bigr]$.


\begin{lemma}\label{lem:one-step-drop}
Let $S\subseteq Q_{m,b}$ be a finite-dimensional $\F_q$-linear space, let $W\subseteq F^b$ be an $F$-linear subspace. Fix $i\in[n]$, and let $P\in \Omega$ be any rational place. Let $S' = \{f\in S : (\psi_i(f))(P)=0\},$ and then we have $\mathrm{dim}_{\F_q}(S\cap W)-\mathrm{dim}_{\F_q}(S'\cap W)\ \le\ \mathrm{dim}_F(\psi_i(W)).$
\end{lemma}

\begin{proof}
Consider the $\F_q$-linear map $\phi_{i,P}:S\cap W \to \F_q^{b_i}$ with $\phi_{i,P}(f)= (\psi_i(f))(P).$ By definition, $\mathrm{ker}(\phi_{i,P})=S'\cap W$, hence $\mathrm{dim}_{\F_q}(S\cap W)-\mathrm{dim}_{\F_q}(S'\cap W)\le \mathrm{dim}_{\F_q}(\phi_{i,P}(S\cap W)).$

As $S\subseteq Q_{m,b}$ and $\deg_{P_\infty}(\psi_i)\le h$, we have $\psi_i(S\cap W)\subseteq \mathcal{L}((m+h)P_\infty)^{b_i}$ and $\phi_{i,P}(S\cap W)=\psi_i(S\cap W)(P)$ entrywise. Therefore, $\mathrm{dim}_{\F_q}(\phi_{i,P}(S\cap W))=\mathrm{dim}_{\F_q}(\psi_i(S\cap W)(P))$. Applying Lemma~\ref{lem:specialization} to the $\F_q$-linear space $\psi_i(S\cap W)$ yields
\[
\mathrm{dim}_{\F_q}(\psi_i(S\cap W)(P))\le \mathrm{dim}_F(\mathrm{span}_F(\psi_i(S\cap W)))\le \mathrm{dim}_F(\psi_i(W)),
\]
which completes this proof.
\end{proof}

\begin{defn}\label{def:potential}
Let $S\subseteq Q_{m,b}$ be a finite dimensional $\F_q$-linear space, let $W\subseteq F^b$ be an $F$-linear subspace and let $\Psi=(\psi_1,\dots,\psi_n)$ be a $b$-local $P_\infty$-bounded profile. Define $d_W(S):= \mathrm{dim}_{\F_q}(S\cap W)$ and $\Pi_W(S,\Psi):= d_W(S)-\sum_{i=1}^n \dim_F(\psi_i(W)).$
\end{defn}

Let $\mathbf{P}=(P_1,\dots,P_n)$ be sampled from $\Omega$ without replacement. For $1\le i\le n$, write $\mathbf{P}_{1:i}:= (P_1,\dots,P_i)$ and $\Psi_{1:i}:= (\psi_1,\dots,\psi_i)$. We define $S_0=S$ and $S_i := S[\Psi_{1:i},\mathbf{P}_{1:i}]$ for $1\leq i\leq n$. A straightforward calculation implies that $S=S_0\supseteq S_1\supseteq \cdots \supseteq S_n=S[\Psi,\mathbf{P}]$.

\begin{lemma}\label{lem:potential-lower}
For every $F$-linear $W\subseteq F^b$ and every realization of $\mathbf{P}$, we have $d_W(S[\Psi,\mathbf{P}])\ \ge\ \Pi_W(S,\Psi).$ In particular, if $\Pi_W(S,\Psi)>0$ then $\dim_{\F_q}(S[\Psi,\mathbf{P}]\cap W)>0$ and $S[\Psi,\mathbf{P}]\neq\{0\}$.
\end{lemma}

\begin{proof}
By Lemma~\ref{lem:one-step-drop}, for each $i\in [n]$, we have $d_W(S_{i-1})-d_W(S_i)\le \mathrm{dim}_F(\psi_i(W)).$ Summing over $i$ yields
\[
\begin{aligned}
d_W(S_n)
&=d_W(S_0)-\sum_{i=1}^n\big(d_W(S_{i-1})-d_W(S_i)\big) \\
&\ge d_W(S)-\sum_{i=1}^n \mathrm{dim}_F(\psi_i(W)) \\
&= d_W(S)-\sum_{i=1}^n \dim_F(\psi_i (W))=\Pi_W(S,\Psi).
\end{aligned}
\]
The proof is completed.
\end{proof}

\begin{lemma}\label{lem:very-negative-implies-dead}
Let $S\subseteq Q_{m,b}$ be a finite-dimensional $\F_q$-linear space and let $\Psi$ be a $b$-local $P_\infty$-bounded profile with $h\ge\deg_{P_\infty}(\Psi)$. Let $W\subseteq F^b$ be an $F$-linear subspace with $D:=\mathrm{dim}_F(W)$, and let $L>0$. Assume $\Pi_W(S,\Psi)\ \le\ -L.$ Then, sampling $\mathbf P=(P_1,\ldots,P_n)$ uniformly from $\Omega$ conditioned on being pairwise distinct, for every $0\le L'\le L$, let $E_{L'}$ be the event that at least one of the following holds:
\begin{enumerate}
\item $W$ is $S[\Psi,\mathbf{P}]$-live; or
\item letting $s\in\{0,1,\dots,n\}$ be the smallest index such that $W$ is
$S[\Psi_{1:s},\mathbf{P}_{1:s}]$-dead, we have $\Pi_W\!\bigl(S[\Psi_{1:s},\mathbf{P}_{1:s}],\,\Psi_{s+1:n}\bigr)\ge -L'.$
\end{enumerate}
Then
\[
\Pr[E_{L'}]
\ \le\
\binom{n}{\left\lceil\frac{L-L'}{D}\right\rceil}
\left(\frac{(m+h)D}{\,N-n+1\,}\right)^{\left\lceil\frac{L-L'}{D}\right\rceil}.
\]
\end{lemma}

\begin{proof}
If $L'=L$, the bound is trivial. If $W$ is already $S$-dead and $L'<L$, then $E_{L'}$ is empty. We may therefore assume $L'<L$ and that $W$ is $S$-live which implies $D>0$. For $i\in [n]$, define $S_i = S[\Psi_{1:i},\mathbf{P}_{1:i}]=\{f\in S_{i-1} : (\psi_i(f))(P_i)=0\}.$ Using Lemma~\ref{lem:one-step-drop}, for $i\in[n]$, define the slack at step $i$ (relative to $W$) by $\Delta_i\ :=\ \mathrm{dim}_F(\psi_i(W))\ -\bigl(d_W(S_{i-1})-d_W(S_i)\bigr)\ \ge\ 0,$ and trivially $\Delta_i\le \mathrm{dim}_F(\psi_i(W))\le D$.

Fix $i\in[n]$ and condition on $(P_1,\dots,P_{i-1})$. Suppose additionally that $W$ is $S_{i-1}$-live, i.e.\ $\mathrm{span}_F(S_{i-1}\cap W)=W$. Thus, we have $D_i:= \mathrm{dim}_F(\mathrm{span}_F(\psi_i(S_{i-1}\cap W)))=\mathrm{dim}_F(\psi_i(W))\le D.$ Moreover, since $S_{i-1}\subseteq Q_{m,b}$ and $\deg_{P_\infty}(\psi_i)\le h$, we have $\psi_i(S_{i-1}\cap W)\subseteq \mathcal{L}((m+h)P_\infty)^{b_i}$. If $\Delta_i\ge 1$, the constraint at step $i$ causes a strict rank loss on $S_{i-1}\cap W$, which implies $\mathrm{dim}_{\F_q}(\psi_i(S_{i-1}\cap W)(P_i))<D_i$. By the specialization Lemma~\ref{lem:specialization} applied to the $\F_q$-space $\psi_i(S_{i-1}\cap W)$, the set of places $P\in\Omega$ for which $\mathrm{dim}_{\F_q}(\psi_i(S_{i-1}\cap W)(P))<D_i$ has size at most $(m+h)D_i\le (m+h)D$. Since $P_i$ is uniform over the remaining pool of size $N-i+1\ge N-n+1$, we obtain
\begin{equation}
\label{eq:single-increase}
\Pr\bigl[\Delta_i\ge 1 \,\big|\, P_1,\dots,P_{i-1}\ \text{and $W$ is $S_{i-1}$-live}\bigr]
\ \le\ \frac{(m+h)D}{N-n+1}.
\end{equation}
On $E_{L'}$, define
\[
t := 
\begin{cases}
n, & \text{if $W$ is $S_n$-live},\\
s, & \text{otherwise, where $s$ is the first index such that $W$ is $S_s$-dead}.
\end{cases}
\]
By definition, $W$ is $S_{t-1}$-live, and $\Pi_W\!\bigl(S[\Psi_{1:t},\mathbf{P}_{1:t}],\,\Psi_{t+1:n}\bigr)\ge -L'$. Therefore, $\sum_{i=1}^t \Delta_i \;\ge\; L-L'.$ Since each $\Delta_i\le D$, it follows that at least $\left\lceil\frac{L-L'}{D}\right\rceil$ indices in $\{1,\dots,t\}$ satisfy $\Delta_i\ge 1$. Moreover, every such index $i\le t$ satisfies that $W$ is $S_{i-1}$-live. Set $r:=\left\lceil(L-L')/D\right\rceil$. If $r>n$, the event is empty. Otherwise, $E_{L'}$ implies the existence of a set $I=\{i_1<\cdots<i_r\}\subseteq[n]$ such that for all $j\in[r]$, $\Delta_{i_j}\ge 1$ and $W\ \text{is}\ S_{i_j-1}\text{-live}.$ By a union bound over $I$,
\begin{equation}
\label{eq:union-over-I}
\Pr[E_{L'}]\ \le\ \binom{n}{r}\cdot
\max_{I:|I|=r}\Pr\Bigl[\bigwedge_{j=1}^r\bigl(\Delta_{i_j}\ge 1\ \wedge\ W\text{ is }S_{i_j-1}\text{-live}\bigr)\Bigr].
\end{equation}
Fix such an $I=\{i_1<\cdots<i_r\}$. Using iterated conditioning at times $i_1,\dots,i_r$ and applying
\eqref{eq:single-increase} at each time (note that $W$ being $S_{i_j-1}$-live is part of the conditioned event), we obtain
\[
\Pr\Bigl[\bigwedge_{j=1}^r\bigl(\Delta_{i_j}\ge 1\ \wedge\ W\text{ is }S_{i_j-1}\text{-live}\bigr)\Bigr]\ \le\ \left(\frac{(m+h)D}{N-n+1}\right)^r.
\]
Plugging into \eqref{eq:union-over-I} completes this proof.
\end{proof}

To bound the probability that a nonzero solution remains, we will argue by induction on the $F$-dimension of $\mathrm{span}_F(S)$. The next lemma constructs the quotient maps needed for this induction and bounds their degrees.

\begin{lemma}\label{lem:quotient-PsiV}
Let $\Psi=(\psi_1,\dots,\psi_n)$ be a $b$-local $P_\infty$-bounded profile where $\psi_i:F^b\to F^{b_i}$. Let $T\subseteq Q_{m,b}=\mathcal{L}(mP_\infty)^b$ be an $\F_q$-linear space, and let $U\subseteq F^b$ be an $F$-linear subspace which is $T$-live. Write $d:=\dim_F U$. The following hold, with all maps written in standard coordinates.
\begin{enumerate}
    \item There exists a surjective $F$-linear map $\pi_U:F^b\to F^{b-d}$ with $\ker(\pi_U)=U$ and $\deg_{P_\infty}(\pi_U)\le dm$.

    \item For each $i\in[n]$, there exist $P_\infty$-bounded $F$-linear maps $\bar\psi_i:F^{b-d}\to F^{r_i}$ and $\chi_{i,U}:F^{b_i}\to F^{r_i}$ for some integer $r_i\ge0$, such that $\ker(\chi_{i,U})=\psi_i(U)$, $\chi_{i,U}\circ\psi_i=\bar\psi_i\circ\pi_U$ and $ \max\{\deg_{P_\infty}(\bar\psi_i),\deg_{P_\infty}(\chi_{i,U})\} \le b\max\{\deg_{P_\infty}(\psi_i),dm\}.$ Thus, these maps fit into the commutative diagram
    \[
    \begin{tikzcd}
    F^b \arrow[r,"\pi_U"] \arrow[d,"\psi_i"'] & F^{b-d} \arrow[d,"\bar\psi_i"] \\
    F^{b_i} \arrow[r,"\chi_{i,U}"'] & F^{r_i}.
    \end{tikzcd}
    \]

    \item If $P\ne P_\infty$ is a rational place and $x\in F^{b_i}$ is regular at $P$ with $x(P)=0$, then $(\chi_{i,U}(x))(P)=0$.
\end{enumerate}
\end{lemma}

\begin{proof}
If $d=0$, take $\pi_U=\mathrm{id}_{F^b}$, $\chi_{i,U}=\mathrm{id}_{F^{b_i}}$ and $\bar\psi_i=\psi_i$. We therefore assume $d\ge1$.
Choose $\boldsymbol{u}^{(1)},\dots,\boldsymbol{u}^{(d)}\in T\cap U$ which form an $F$-basis of $U$, and let $M_U\in F^{b\times d}$ be the matrix whose $j$th column is $\boldsymbol{u}^{(j)}$. Since each $\boldsymbol{u}^{(j)}$ lies in $\mathcal{L}(mP_\infty)^b$, every entry of $M_U$ lies in $\mathcal{L}(mP_\infty)$. Since $\rank_F(M_U)=d$, choose a permutation matrix $\Sigma$ such that
\[
\Sigma M_U=
\begin{bmatrix}
A\\ B
\end{bmatrix},
\qquad
A\in F^{d\times d}\ \text{invertible},\quad B\in F^{(b-d)\times d}.
\]
Define
\[
M_\pi:=\det(A)\begin{bmatrix}-BA^{-1}&I_{b-d}\end{bmatrix}\Sigma
\in F^{(b-d)\times b}.
\]
Then $\ker(M_\pi)=\mathrm{Im}(M_U)=U$ by the usual block-elimination argument. Hence $\pi_U(x):=M_\pi x$ identifies $F^b/U$ with $F^{b-d}$. Also, each entry of $\mathrm{adj}(A)$ lies in $\mathcal L((d-1)\cdot mP_\infty)$, hence $\det(A)A^{-1}=\mathrm{adj}(A)$ is $(d-1)\cdot m$-bounded at $P_\infty$, and therefore $\det(A)BA^{-1}$ is $d\cdot m$-bounded at $P_\infty$. Thus every entry of $M_\pi$ lies in $\mathcal L(dmP_\infty)$. Fix $i\in[n]$, write $a=b_i$, and let $M_{\psi_i}\in F^{a\times b}$ be the matrix of $\psi_i$ in the given coordinates. Let
\[
d':=\max\{\deg_{P_\infty}(\psi_i),d\cdot m\}, \qquad Z:=
\begin{bmatrix}
M_\pi\\ M_{\psi_i}
\end{bmatrix}
\in F^{(b-d+a)\times b}.
\]
If $Z=0$, take $r_i=a$, $\chi_{i,U}=\mathrm{id}_{F^a}$ and $\bar\psi_i=0$; the claims for this $i$ follow immediately. Otherwise, choose $e:=\rank_F(Z)$ linearly independent columns of $Z$, and let $T\in F^{(b-d+a)\times e}$ be the corresponding submatrix. Choose a permutation matrix $\Sigma^\prime$ so that
\[
\Sigma^\prime T=
\begin{bmatrix}
T_1\\ T_2
\end{bmatrix},
\qquad
T_1\in F^{e\times e}\ \text{invertible}.
\]
Define
\[
T^\star:= \det(T_1)\,
\begin{bmatrix}
-T_2T_1^{-1} & I_{b-d+a-e}
\end{bmatrix}\Sigma^\prime
\in F^{(b-d+a-e)\times (b-d+a)}.
\]
Exactly as above, $\ker(T^\star)=\mathrm{Im}(T)=\mathrm{Im}(Z)$. Since every entry of $T$ lies in $\mathcal{L}(d'P_\infty)$ and $e\le b$, every entry of $T^\star$ lies in $\mathcal{L}(b\cdot d'\cdot P_\infty)$. Write $T^\star=[P\ \ Q]$ in block form, where $P\in F^{r\times (b-d)},\,Q\in F^{r\times a}$ and $r:=b-d+a-e\le a$; set $r_i=r$. Then the relation $T^\star Z=0$ is precisely $PM_\pi+QM_{\psi_i}=0$. Thus $Q$ represents $\chi_{i,U}:F^a\to F^r$ and $-P$ represents $\bar\psi_i:F^{b-d}\to F^r$, with $\chi_{i,U}\circ\psi_i=\bar\psi_i\circ\pi_U$ and the stated degree bounds.

We next prove $\ker(\chi_{i,U})=\psi_i(U)$, where $\chi_{i,U}$ is the map represented by $Q$. If $u\in U=\ker(\pi_U)$, then $M_\pi u=0$, hence $0=T^\star Zu=PM_\pi u+QM_{\psi_i}u=Q\psi_i(u),$ so $\psi_i(U)\subseteq \ker(\chi_{i,U})$. Conversely, if $x\in F^a$ satisfies $\chi_{i,U}(x)=0$, then $Qx=0$, so $y:=(0,x)\in F^{b-d+a}$ lies in $\ker(T^\star)=\mathrm{Im}(Z)$. Hence $y=Zw$ for some $w\in F^b$, and therefore $M_\pi w=0$ and $M_{\psi_i} w=x$. Thus $w\in U$ and $x=\psi_i(w)\in\psi_i(U)$.

Every entry of the matrix of $\chi_{i,U}$ is regular away from $P_\infty$. Hence if $x$ is regular at $P\ne P_\infty$ and satisfies $x(P)=0$, then each coordinate of $\chi_{i,U}(x)$ is an $F$-linear combination of the coordinates of $x$ with coefficients regular at $P$, and therefore $(\chi_{i,U}(x))(P)=0$.
\end{proof}

\begin{lemma}\label{lem:probably-trivial}
Let $\Psi=(\psi_1,\dots,\psi_n)$ be a $b$-local $P_\infty$-bounded profile.  Let $S\subseteq Q_{m,b}=\mathcal{L}(mP_\infty)^b$ be a finite-dimensional $\F_q$-linear space. Let $W=\mathrm{span}_F(S)$ with $D_W=\dim_F(W)$ and $h\ge\max\{\deg_{P_\infty}(\Psi),\,bm\}$. 

Suppose there exists $\lambda\ge 2D_W\cdot(D_W+1)$
such that for every $S$-live $F$-linear subspace $U\subseteq W$,
$\Pi_U(S,\Psi) \le -\lambda\cdot \mathrm{dim}_F(U)$.
If $S=\{0\}$, the probability below is zero. Otherwise, if $\mathbf{P}=(P_1,\ldots,P_n)$ is sampled uniformly from $\Omega$ without replacement,
\[
\Pr\big[S[\Psi,\mathbf{P}]\neq\{0\}\big]
\ \le\
(2^{D_W}-1)\cdot
\left(\frac{4e\cdot b^{2D_W}\cdot h\cdot n}{\lambda\cdot (N-n+1)}\right)^{\frac{\lambda}{2D_W}}.
\]
\end{lemma}

\begin{proof}
We proceed with induction on $D_W$. If $D_W=0$, then $W=\{0\}$, hence $S=\{0\}$ and therefore $S[\Psi,\mathbf P]=\{0\}$ for every $\mathbf P$. Assume now that $D_W\ge1$ and that the lemma holds for every smaller $F$-span dimension, for all choices of the remaining parameters. If the base on the right hand side is at least $1$, the claim is trivial. We therefore assume it is less than $1$. For $0\le i\le n$, define $S_i:=S[\Psi_{1:i},\mathbf P_{1:i}]$ and $p_{U,i}:=\Pi_U(S_i,\Psi_{i+1:n})$ for every $F$-linear subspace $U\subseteq W$, where $\Psi_{1:i}:=(\psi_1,\dots,\psi_i)$ and $\Psi_{i+1:n}:=(\psi_{i+1},\dots,\psi_n).$ By Lemma~\ref{lem:one-step-drop}, for fixed $U$ the sequence $p_{U,0}\le p_{U,1}\le \cdots\le p_{U,n}$ is nondecreasing. Let 
\[
L:=\lambda D_W,
\qquad
\lambda^\star:=\left\lceil \frac{(D_W-1)\lambda}{D_W}+D_W\right\rceil,
\qquad
L':=\lambda^\star D_W.
\]
Let $t\in\{1,\ldots,n\}$ be the least index such that at least one of the following occurs:
\begin{enumerate}
    \item $W$ is $S_t$-dead;
    \item there exists a nonzero proper $S_t$-live subspace $U\subsetneq W$ with
    $p_{U,t}>-\lambda^\star\cdot \dim_F U.$
\end{enumerate}
If no such index exists, set $t=\infty$. Define the event $G:=\{\,t<\infty \text{ and } p_{W,t}< -L'\,\}.$ If $t=\infty$, then $W$ is $S_n$-live. Otherwise, on $\overline G$, we have $p_{W,t}\ge-L'$, and this remains true until $W$ first becomes dead, if it does. Thus $\overline G\subseteq E_{L'}$. Applying Lemma~\ref{lem:very-negative-implies-dead} with parameters $(L,L')$ yields
\[
\Pr[\overline{G}]
\le
\binom{n}{\left\lceil \frac{L-L'}{D_W} \right\rceil}
\left(
\frac{(m+h)D_W}{N-n+1}
\right)^{\left\lceil \frac{L-L'}{D_W} \right\rceil}.
\]
Writing $r=\left\lceil\frac{L-L'}{D_W}\right\rceil,$ we have $r\ge \lambda/(2D_W)$. Using $m+h\le 2h$, $\binom nr\le (en/r)^r$ and the fact that $\frac{4e\,h\,nD_W^2}{\lambda(N-n+1)}<1,$ we obtain
\[
\Pr[\overline G]
\le
\left(
\frac{4e\,h\,nD_W^2}
{\lambda(N-n+1)}
\right)^{\lambda/(2D_W)}.
\]

If $\Pr[G]=0$, the bound above proves the claim. Otherwise, fix a prefix $(t,P_1,\dots,P_t)$ for which $G$ holds. All probabilities below are conditional on this fixed prefix, so the remaining places are uniform without replacement from the remaining $N-t$ places. If $S_t=\{0\}$, then $S_n=\{0\}$. This always occurs on $G$ when $D_W=1$. We may therefore assume $D_W\ge2$ and $S_t\ne\{0\}$. Choose a nonzero proper subspace $U\subsetneq W$ as follows:
\begin{itemize}
    \item If $W$ is $S_t$-dead, set $U:=\mathrm{span}_F(S_t)$.
    \item Otherwise choose $U\subsetneq W$ maximal with respect to inclusion among all nonzero proper $S_t$-live subspaces satisfying $p_{U,t}>-\lambda^\star\cdot \dim_F U.$
\end{itemize}
In either case, $1\le D_U:=\dim_F U\le D_W-1$. We split the proof into two parts:
\[
(S_n\cap U=\{0\}) \text{ and } (S_n\subseteq U)
\qquad \Longrightarrow \qquad
S_n=\{0\}.
\]
\paragraph{Part I: controlling $S_n\cap U$.}
Let $S':=S_t\cap U$ and $\Psi':=\Psi_{t+1:n}.$ Then $\mathrm{span}_F(S')=U$ and $S_n\cap U=S'[\Psi',\mathbf P_{t+1:n}]$.

Let $V\subseteq U$ be an $F$-linear subspace which is $S'$-live. Since $V\subseteq U\subsetneq W$, the stopping condition had not yet been triggered by $V$ at time $t-1$, and therefore $p_{V,t-1}\le -\lambda^\star\cdot \dim_F V$. By Lemma~\ref{lem:one-step-drop}, we have $p_{V,t}\le p_{V,t-1}+\dim_F V
\le -(\lambda^\star-1)\cdot \dim_F V.$ Also, $\lambda^\star-1
\ge \frac{(D_W-1)\lambda}{D_W}
\ge 2(D_W-1)(D_W+1)
\ge 2D_U(D_U+1).$ Therefore, by the induction hypothesis applied to $S'\subseteq F^b$,
\[
\Pr_{\mathbf{P}}[S_n\cap U\neq \{0\}\mid P_1,\ldots,P_t]
\le
(2^{D_U}-1)\left(
\frac{4e\cdot b^{2D_U}\cdot h\cdot n}
{(\lambda^\star-1)\cdot (N-n+1)}
\right)^{\frac{\lambda^\star-1}{2D_U}}.
\]
\medskip
\noindent
\textbf{Part II: controlling $S_n\subseteq U$.}
If $W$ is $S_t$-dead, then by construction $S_n\subseteq S_t\subseteq U$ and we have nothing to prove. We now assume that $W$ is $S_t$-live.

Applying Lemma~\ref{lem:quotient-PsiV}, we obtain an $F$-linear quotient map $\pi:F^b\to F^{b-D_U}$ with $\ker\pi=U$ and $\deg_{P_\infty}(\pi)\le D_U\cdot m$. Moreover, for each $i\in\{t+1,\dots,n\}$, we obtain $F$-linear maps
\[
\bar\psi_i:F^{b-D_U}\to F^{r_i},
\qquad
\chi_i:F^{b_i}\to F^{r_i}
\]
such that $\chi_i\circ \psi_i=\bar\psi_{i}\circ \pi$, $\ker(\chi_i)=\psi_i(U)$ and $\deg_{P_\infty}(\bar\psi_i)\le b\max\{\deg_{P_\infty}(\psi_i),D_Um\}\le bh$.

\noindent Let
\[
S':=\pi(S_t),
\qquad
\Psi':=(\bar\psi_{t+1},\dots,\bar\psi_n),
\qquad
W':=\mathrm{span}_F(S')=\pi(W).
\]
Since $D_U\ge 1$ and $\ker(\pi)\cap W=U$, we have $\dim_F W' = D_W-D_U < D_W$.

We first claim that
\[
S'[\Psi',\mathbf{P}_{t+1:n}]=\{0\}
\quad\Longrightarrow\quad
S_n\subseteq U.
\]
Indeed, let $f\in S_n$. Then for every $i>t$, $(\psi_i(f))(P_i)=0$. Using Lemma~\ref{lem:quotient-PsiV}, we obtain $(\bar\psi_i(\pi(f)))(P_i)=0$ for $i=t+1,\dots,n$. Hence $\pi(f)\in S'[\Psi',\mathbf{P}_{t+1:n}]=\{0\}$. Therefore, $f\in \ker(\pi)=U$.

It remains to verify the potential hypothesis for $(S',\Psi')$ inside $W'$. The case $V'=\{0\}$ is immediate. Let $V'\subseteq W'$ be any nonzero $S'$-live subspace, and define $U':=\pi^{-1}(V')\cap W$. Then $U\subseteq U'\subseteq W$ and $\dim_F V'=\dim_F U'-\dim_F U$. We claim that the sequence
\[
0\to S_t\cap U \to S_t\cap U' \xrightarrow{\ \pi\ } S'\cap V' \to 0
\]
is exact. Injectivity on the left is obvious. To see surjectivity on the right, let $\boldsymbol{y}\in S'\cap V'$. As $\boldsymbol{y}\in \pi(S_t)$, choose $\boldsymbol{x}\in S_t$ with $\pi(\boldsymbol{x})=\boldsymbol{y}$. Since $\boldsymbol{y}\in V'$, we also have $\boldsymbol{x}\in \pi^{-1}(V')\cap W=U'$. Thus, $\boldsymbol{x}\in S_t\cap U'$ and $\pi(\boldsymbol{x})=\boldsymbol{y}$. Consequently,
\[
\dim_{\F_q}(S'\cap V') = \dim_{\F_q}(S_t\cap U')-\dim_{\F_q}(S_t\cap U).
\]
For each $i>t$, the sequence
\[
0\to \psi_i(U)\to \psi_i(U') \xrightarrow{\ \chi_i\ } \bar\psi_i(V')\to 0
\]
is exact because the first map is inclusion, $\ker(\chi_i)=\psi_i(U)$, and $\bar\psi_i(V')=\bar\psi_i(\pi(U'))=\chi_i(\psi_i(U')).$ Hence, $\dim_F \bar\psi_i(V')=\dim_F \psi_i(U')-\dim_F \psi_i(U)$. Substituting these two identities into the definition of $\Pi_{V'}$ gives
\begin{equation}\label{eq:potential-V}
    \Pi_{V'}(S',\Psi')=p_{U',t}-p_{U,t}.
\end{equation}

We now show that $U'$ is $S_t$-live. Since $V'$ is $S'$-live, we have $V'=\mathrm{span}_F(S'\cap V')=\pi(\mathrm{span}_F(S_t\cap U')).$ Therefore, $U'=\pi^{-1}(V')\cap W = \mathrm{span}_F(S_t\cap U') + U.$ As $U$ is $S_t$-live by construction, $U\subseteq \mathrm{span}_F(S_t\cap U')$ and $U'=\mathrm{span}_F(S_t\cap U')$, proving that $U'$ is $S_t$-live. By maximality of $U$, $p_{U',t}\le-\lambda^\star\dim_F U'$ when $U'\subsetneq W$; if $U'=W$, the same inequality follows from $G$. Since $p_{U,t}>-\lambda^\star\dim_F U$, Eq.~\eqref{eq:potential-V} gives
\[
\Pi_{V'}(S',\Psi')\le-\lambda^\star(\dim_F U' - \dim_F U)=-\lambda^\star \dim_F V'.
\]
By the definition of $\lambda^{\star}$,
\[
\lambda^\star\ge\frac{(D_W-1)\lambda}{D_W}\ge\frac{2D_W(D_W+1)(D_W-1)}{D_W}\ge 2(D_{W'}+1)D_{W'}.
\]
Since $\deg_{P_\infty}(\pi)\le D_U\cdot m$ and $S_t\subseteq Q_{m,b}$, every element of $S'=\pi(S_t)$ lies in $\mathcal{L}((D_U+1)mP_\infty)^{\,b-D_U}$. Thus, $S'$ is $P_{\infty}$-bounded and $\deg_{P_\infty}(\Psi')\le bh$. 

Since $D_U+1\le b$ and $h\ge bm$, we may take $h':=bh\ge\max\{\deg_{P_\infty}(\Psi'),b(D_U+1)m\}$. Applying the induction hypothesis to $(S',\Psi')$ and $W'$, and using $b-D_U\le b$, yields
\[
\begin{aligned}
\Pr_{\mathbf{P}}[S_n \nsubseteq U \mid P_1,\ldots,P_t]
&\le \Pr_{\mathbf{P}}\bigl[S'[\Psi',\mathbf P_{t+1:n}] \ne \{0\}\mid P_1,\ldots,P_t\bigr] \\
&\le
(2^{D_{W'}}-1)
\left( \frac{4e\cdot b^{2D_{W'}}\cdot b \cdot h\cdot (n-t)} {(N-n+1)\cdot \lambda^\star}
\right)^{\frac{\lambda^\star}{2D_{W'}}}.
\end{aligned}
\]
If $W$ is $S_t$-dead, the bound from Part I alone suffices. Otherwise, combining the two bounds gives
\[
\begin{aligned}
\Pr_{\mathbf{P}}[S_n \ne \{0\}\mid P_1,\ldots,P_t]
&\le
\Pr_{\mathbf{P}}[S_n \cap U \ne \{0\}\mid P_1,\ldots,P_t] \\
&\quad+\Pr_{\mathbf{P}}[S_n \nsubseteq U\mid P_1,\ldots,P_t] \\
&\le
(2^{D_U}-1)
\left(
\frac{4e\cdot b^{2D_U}\cdot h\cdot n}
{(N-n+1)\cdot (\lambda^\star-1)}
\right)^{\frac{\lambda^\star-1}{2D_U}} \\
&\quad+
(2^{D_{W'}}-1)
\left(
\frac{4e\cdot b^{2D_{W'}+1}\cdot h\cdot (n-t)}
{(N-n+1)\cdot \lambda^\star}
\right)^{\frac{\lambda^\star}{2D_{W'}}} \\
&\le
(2^{D_W}-2)
\left(
\frac{4e\cdot b^{2D_W}\cdot h\cdot n}
{(N-n+1)\cdot \lambda}
\right)^{\frac{\lambda}{2D_W}}.
\end{aligned}
\]
For the last inequality, $D_U+D_{W'}=D_W$, and
\[
\frac{\lambda}{\lambda^\star-1},\ \frac{\lambda}{\lambda^\star}
\le\frac{D_W}{D_W-1}\le b.
\]
Thus both bases are at most the base on the right hand side, which is less than $1$, and both exponents are at least $\lambda/(2D_W)$. Also,
$(2^{D_U}-1)+(2^{D_{W'}}-1)\le2^{D_W}-2$.
Averaging over the fixed prefixes gives the same bound conditional on $G$.

Recall that we already have an upper bound on $\Pr[\overline{G}]$. Therefore,
\[
\begin{aligned}
\Pr_{\mathbf{P}}[S_n \ne \{0\}]
&=
\Pr_{\mathbf{P}}[\overline{G}\wedge (S_n \ne \{0\})]
+
\Pr_{\mathbf{P}}[G\wedge (S_n \ne \{0\})] \\
&\le
\Pr_{\mathbf{P}}[\overline{G}]
+
\Pr_{\mathbf{P}}[S_n \ne \{0\}\mid G] \\
&\le
\left(
\frac{4e\cdot D_W^2\cdot h\cdot n}
{(N-n+1)\cdot \lambda}
\right)^{\frac{\lambda}{2D_W}}
+
(2^{D_W}-2)
\left(
\frac{4e\cdot b^{2D_W}\cdot h\cdot n}
{(N-n+1)\cdot \lambda}
\right)^{\frac{\lambda}{2D_W}} \\
&\le
(2^{D_W}-1)
\left(
\frac{4e\cdot b^{2D_W}\cdot h\cdot n}
{(N-n+1)\cdot \lambda}
\right)^{\frac{\lambda}{2D_W}},
\end{aligned}
\]
which yields the lemma.
\end{proof}

We now apply Lemma~\ref{lem:probably-trivial} to $S=Q_{m,b}$. Lemma~\ref{lem:lcl-to-Pinf} allows us to verify the required bound on $\Pi$ using $\Phi$.

\begin{thm}\label{thm:C_contain_Pr}
Let $b\in\N$ and let $q>\binom{b}{2}$ be a prime power. Let $\V$ be a $b$-local profile, and let $\cC$ be a random one-point AG code of rate $R$ as in Definition~\ref{def:random-one-point-ag}, with $0\le m<n\le N$. Suppose that $\varepsilon>0$ satisfies $\Phi(\V,U,R)\le-\varepsilon n\dim_{\F_q}U$ for every nonzero $U\in\mathcal{L}(\F_q^b)$.
If $\varepsilon n\ge2b(b+1)$, then
\[
\Pr\bigl[
\cC\text{ contains a nonzero matrix satisfying }\V
\bigr]
\le
(2^b-1)
\left(
\frac{(4b)^{4b}m}{\varepsilon(N-n+1)}
\right)^{\varepsilon n/(2b)}.
\]
\end{thm}

\begin{proof}
Let $S = Q_{m,b}$ and let $\Psi = \Psi_{\mathcal{V}} = (\psi_1,\ldots,\psi_n)$ be the $b$-local $P_\infty$-bounded profile associated with $\mathcal{V}$. Fix a nontrivial $Q_{m,b}$-live $F$-linear subspace $U \subseteq F^b$. Lemma~\ref{lem:lcl-to-Pinf}(2) gives an $\F_q$-linear subspace $U'\subseteq\F_q^b$ with
$\dim_{\F_q}U'=\dim_F U$ such that
\[
\begin{aligned}
\sum_{i=1}^n \dim_F \psi_i(U) \ge \sum_{i=1}^n \Bigl(\dim_F U - \dim_{\mathbb{F}_q}(U' \cap V_i)\Bigr) =n \cdot \dim_F U - \sum_{i=1}^n \dim_{\mathbb{F}_q}(U' \cap V_i).
\end{aligned}
\]
We rewrite the last term via the potential function $\Phi(\mathcal{V},U',R)$ by $\Phi(\mathcal{V},U',R)=\sum_{i=1}^n \dim_{\mathbb{F}_q}(U' \cap V_i)-(n-k)\dim_{\mathbb{F}_q} U'.$ Therefore, by assumption of $\Phi(\mathcal{V},U',R) \le -\varepsilon \cdot n\dim_{\mathbb{F}_q} U'=-\varepsilon \cdot n\dim_{F} U$,
\[
\sum_{i=1}^n \dim_F \psi_i(U) \ge k \dim_F U - \Phi(\mathcal{V},U',R)\ge (k+\varepsilon n)\dim_F U.
\]
By definition, $\Pi_U(Q_{m,b},\Psi_{\mathcal{V}})=\dim_{\mathbb{F}_q}(U \cap Q_{m,b})-\sum_{i=1}^n \dim_F \psi_i(U).$ We claim that $\dim_{\mathbb{F}_q}(U \cap Q_{m,b}) \le k \cdot \dim_F U.$

To see this, let $d := \dim_F U$ and choose a generating matrix $G \in F^{b \times d}$ for $U$ in systematic form, so that $U = \{Gx : x \in F^d\}$ and $G$ contains a $d \times d$ identity submatrix. Then every element of $U \cap Q_{m,b}$ can be written as $Gx$ for some $\boldsymbol{x} \in Q_{m,d}$, and thus $|U \cap Q_{m,b}| \le |Q_{m,d}| = q^{kd}.$ Therefore $\dim_{\mathbb{F}_q}(U \cap Q_{m,b})= \log_q |U \cap Q_{m,b}| \le kd,$ proving the claim. Combining the above two inequalities, we obtain
\[
\Pi_U(Q_{m,b},\Psi_{\mathcal{V}})\le k\dim_F U - (k+\varepsilon n)\dim_F U= -\varepsilon n \cdot \dim_F U.
\]

Since $\mathrm{span}_F(S)=F^b$ and $\deg_{P_\infty}(\Psi_\V)=0$, apply Lemma~\ref{lem:probably-trivial} with $\lambda=\varepsilon n$, $D_W=b$, and $h=bm$. Its hypotheses hold by the inequality above and $\varepsilon n\ge2b(b+1)$. Since $4e b^{2b+1}\le(4b)^{4b}$, we obtain
\[
\Pr\bigl[\dim_{\F_q}(S[\Psi,\mathbf{P}])>0\bigr]
\le
(2^b-1)
\left(
\frac{(4b)^{4b}m}{\varepsilon(N-n+1)}
\right)^{\varepsilon n/(2b)}.
\]
Finally, Lemma~\ref{lem:lcl-to-Pinf}(1) shows that any nonzero matrix in $\cC$ satisfying $\V$ has a nonzero preimage in $S[\Psi,\mathbf{P}]$. The desired bound follows.
\end{proof}

When $R\le R_{\mathcal P}-\varepsilon$, we will apply Theorem~\ref{thm:C_contain_Pr} to the profile obtained by projecting each $V_i$ to $\F_q^b/W$. The following lemma helps us choose $W$ so that the required bound on $\Phi$ holds.

\begin{lemma}[{\cite[Lemma~4.5]{levi2024random}}]
\label{lem:4.5}
Let $b\ge2$ and let $\V$ be a $b$-local profile. There exists $W\in\mathcal L(\F_q^b)\setminus \mathcal L_{\mathrm{Dist}}(\F_q^b)$ such that $\Phi(\V,W,R_\V)=\max_{U\in\mathcal L(\F_q^b)}\Phi(\V,U,R_\V).$
\end{lemma}

\begin{proof}[Proof of Theorem~\ref{thm:threshold-AG-codes}, part~1]
Fix $\V\in\mathcal F$. Since $R\le R_{\mathcal P}-\varepsilon$ and $R_{\mathcal P}\le R_\V$, we have $R_\V-R\ge\varepsilon$. By Lemma~\ref{lem:4.5}, choose $W\notin\mathcal L_{\mathrm{Dist}}(\F_q^b)$ maximizing $\Phi(\V,\cdot,R_\V)$. Put $d=\dim W<b$, and choose a surjective $\F_q$-linear map $\varphi:\F_q^b\longrightarrow\F_q^{b-d}$ with $\ker\varphi=W.$ Let $\V'=(\varphi(V_1),\ldots,\varphi(V_n))$.

Suppose that $\cC$ contains $\V$, and let $A$ be a matrix satisfying $\V$ whose columns are pairwise distinct and belong to $\cC$. Applying $\varphi$ to each row of $A$ gives an $n\times(b-d)$ matrix $A'$ satisfying $\V'$. Its columns belong to $\cC$, since they are linear combinations of the columns of $A$. Moreover, $A'\ne0$: otherwise all rows of $A$ would lie in $W$, which would force two columns of $A$ to be equal because $W\notin\mathcal L_{\mathrm{Dist}}(\F_q^b)$. Consequently,
\[
\Pr[\cC\text{ contains }\V]
\le
\Pr[\cC\text{ contains a nonzero matrix satisfying }\V'].
\]

Let $U'\subseteq\F_q^{b-d}$ be a nonzero subspace and put $U=\varphi^{-1}(U')$. Then $W\subseteq U$ and $\dim U'=\dim U-d$. By the choice of $W$,
\[
\begin{aligned}
\Phi(\V,U,R)-\Phi(\V,W,R)
&=\Phi(\V,U,R_\V)-\Phi(\V,W,R_\V) -(R_\V-R)n(\dim U-d)\\
&\le -\varepsilon n(\dim U-d).
\end{aligned}
\]
Since $\varphi(V_i)\cap U'=\varphi(V_i\cap U)$ and $\ker\varphi=W$, we also have $\dim(\varphi(V_i)\cap U')=\dim(V_i\cap U)-\dim(V_i\cap W).$ It follows that
\[
\begin{aligned}
\Phi(\V',U',R)
&=\sum_{i=1}^n
  \bigl(\dim(V_i\cap U)-\dim(V_i\cap W)\bigr)-(1-R)n(\dim U-d)\\
&=\Phi(\V,U,R)-\Phi(\V,W,R)\le -\varepsilon n\dim U'.
\end{aligned}
\]
Thus Theorem~\ref{thm:C_contain_Pr} applies to $\V'$ with local size $b-d$, since $1\le b-d\le b$ and its remaining assumptions follow from those for $b$. Therefore,
\[
\Pr[\cC\text{ contains }\V]
\le
(2^{b-d}-1)
\left(
\frac{(4(b-d))^{4(b-d)}m}
{\varepsilon(N-n+1)}
\right)^{\varepsilon n/(2(b-d))}.
\]
If $(4b)^{4b}m/(\varepsilon(N-n+1))\ge1$, the desired bound for this fixed profile is trivial. Otherwise, increasing the base and decreasing the exponent gives
\[
\Pr[\cC\text{ contains }\V]
\le
(2^b-1)
\left(
\frac{(4b)^{4b}m}{\varepsilon(N-n+1)}
\right)^{\varepsilon n/(2b)}.
\]
A union bound over $\V\in\mathcal F$ proves part~1.
\end{proof}

\subsection{Above the random linear codes threshold}
\label{subsec:above-threshold}

We work in the setting of Theorem~\ref{thm:threshold-AG-codes} and write $k=\ell(mP_\infty)=Rn$. Consider the $\mathbb F_q$-linear space of matrices $A\in\mathcal C^b$ satisfying the profile $\mathcal V$. For each pair $u<v$, let the collision subspace consist of those matrices whose $u$-th and $v$-th columns are equal. We first estimate the dimension of the full space and then show that every collision subspace has strictly smaller dimension. Since there are only $\binom{b}{2}$ such subspaces and $q>\binom{b}{2}$, their union cannot cover the full space. Hence there is a matrix satisfying $\mathcal V$ whose columns are pairwise distinct. Throughout this subsection, all dimensions are over $\mathbb F_q$.

Fix a $b$-local profile $\mathcal V$. For an $\F_q$-linear space $W\subseteq\mathbb F_q^b$, let
\[
K_W:=\{A\in\mathcal C^b:
A_{i,*}\in V_i\cap W\text{ for every }i\in[n]\},
\qquad
c(W):=\sum_{i=1}^n
\bigl(\dim W-\dim(V_i\cap W)\bigr).
\]
Thus $\Phi(\mathcal V,W,R)=k\dim W-c(W).$


\begin{lemma}
\label{lem:ag-constrained-dimension}
Let $\lambda\ge 2b(b+1)$. For every fixed $\F_q$-linear space $W\subseteq\mathbb F_q^b$,
\[
\Pr\!\left[
\dim K_W>
\max_{T\subseteq W}\Phi(\mathcal V,T,R)
+(\lambda+1)\dim W
\right]
\le b(2^b-1)
\left(
\frac{(4b)^{4b}mn}{\lambda(N-n+1)}
\right)^{\lambda/(2b)}.
\]
\end{lemma}

\begin{proof}

The case $W=\{0\}$ is immediate. We therefore assume $W\neq\{0\}$ and fix it throughout the proof. Starting from $W_0=\{0\}$, suppose that $W_{j-1}\neq W$ has already been chosen. Among all subspaces $T$ satisfying $W_{j-1}\subsetneq T\subseteq W,$ choose $W_j$ to be a subspace $T$ for which
\[
\alpha(T\mid W_{j-1}):=
\frac{c(T)-c(W_{j-1})}{\dim T-\dim W_{j-1}}
\]
is minimal and write $a_j:=\alpha(W_j\mid W_{j-1}),\,d_j:=\dim W_j-\dim W_{j-1}.$ Since the dimension strictly increases at every step, this process terminates after at most $\dim W\le b$ steps. Thus, for some $s\le\dim W$, we obtain the sequence $\{0\}=W_0\subsetneq W_1\subsetneq\cdots\subsetneq W_s=W$ and $c(W_j)-c(W_{j-1})=a_jd_j.$

The numbers $a_1,\ldots,a_s$ are nondecreasing. Indeed, fix $j<s$ and let $W_j\subsetneq T\subseteq W$. By the choice of $W_j$, $c(T)-c(W_{j-1})\ge a_j\bigl(\dim T-\dim W_{j-1}\bigr).$ Subtracting $c(W_j)-c(W_{j-1})=a_j\bigl(\dim W_j-\dim W_{j-1}\bigr)$ gives $\frac{c(T)-c(W_j)} {\dim T-\dim W_j}\ge a_j.$ Taking the minimum over such $T$ gives $a_{j+1}\ge a_j$. Hence $a_1\le a_2\le\cdots\le a_s.$

We now consider each extension $W_{j-1}\subsetneq W_j$ separately. Fix $j\in[s]$ and choose a surjective linear map $\pi_j:W_j\longrightarrow\F_q^{d_j}$ with $\ker(\pi_j)=W_{j-1}$. For each $i\in[n]$, define $V_i^{(j)}:=\pi_j(V_i\cap W_j)$ and $\V^{(j)}:=(V_1^{(j)},\ldots,V_n^{(j)}).$ Let $Z\subseteq\F_q^{d_j}$ and let $T:=\pi_j^{-1}(Z)$. Then $\dim Z=\dim T-\dim W_{j-1}$ and $\dim(V_i^{(j)}\cap Z)=\dim(V_i\cap T)-\dim(V_i\cap W_{j-1}).$ Therefore,
\begin{equation}
\label{eq:layer-constraint-lower-bound}
\sum_{i=1}^n
\bigl(
\dim Z-\dim(V_i^{(j)}\cap Z)
\bigr)
=
c(T)-c(W_{j-1})
\ge a_j\dim Z,
\end{equation}
where the last inequality follows from the definition of $a_j$ when $Z\neq\{0\}$ and is trivial for $Z=\{0\}$.

Set $k_j:=\min\{k,\max\{0,\lfloor a_j-\lambda\rfloor\}\}.$ If $k_j=0$, let $\cC_j=\{0\}$. Otherwise, choose an integer $0\le m_j\le m$ such that $\ell(m_jP_\infty)=k_j$ and set $\cC_j:=\mathrm{ev}_{\mathbf P}\mathcal L(m_jP_\infty)\subseteq\cC.$ Such an $m_j$ exists since $\ell(tP_\infty)$ increases from $1$ to $k$ in steps of size at most one as $t$ goes from $0$ to $m$. The chain, the maps $\pi_j$, and the integers $m_j$ are all fixed before the places are sampled. Assume now that $k_j>0$. For every nonzero $Z\subseteq\F_q^{d_j}$,~\eqref{eq:layer-constraint-lower-bound} gives $\Phi(\V^{(j)},Z,k_j/n)\le(k_j-a_j)\dim Z\le-\lambda\dim Z.$ Hence, when $k_j>0$, Theorem~\ref{thm:C_contain_Pr} applied to $\cC_j$ and $\V^{(j)}$ gives
\[
\Pr\!\left[
\cC_j^{d_j}
\text{ contains a nonzero matrix satisfying }\V^{(j)}
\right]
\le (2^b-1)
\left(
\frac{(4b)^{4b}mn}{\lambda(N-n+1)}
\right)^{\lambda/(2b)},
\]
using $d_j\le b$, $m_j\le m$. The size conditions hold since $\lambda\ge2b(b+1)\ge2d_j(d_j+1)$. For $k_j=0$, this event is impossible.

Define $K^{(j)}:=\{A\in\cC^{d_j}:A_{i,*}\in V_i^{(j)}\text{ for every }i\in[n]\}.$ Outside the event above, $K^{(j)}\cap\cC_j^{d_j}=\{0\}.$ Since $\dim\cC^{d_j}=kd_j$ and $\dim\cC_j^{d_j}=k_jd_j,$ we obtain $\dim K^{(j)}\le(k-k_j)d_j.$ Applying $\pi_j$ to every row of a matrix in $K_{W_j}$ gives a linear map $K_{W_j}\longrightarrow K^{(j)}$ whose kernel is exactly $K_{W_{j-1}}$. Its image has columns in $\cC$, since these columns are $\F_q$-linear combinations of the original columns. Hence $\dim K_{W_j}-\dim K_{W_{j-1}}\le(k-k_j)d_j.$ Summing over $j\in[s]$ gives $\dim K_W\le\sum_{j=1}^s(k-k_j)d_j.$ By the definition of $k_j$, we have $k-k_j\le(k-a_j)_++\lambda+1$ where $(x)_+:=\max\{x,0\}.$ 

Since $a_1\le\cdots\le a_s$, let $\ell$ be the last index such that $a_\ell\le k$, with $\ell=0$ if no such index exists. Then
\[
\sum_{j=1}^s d_j(k-a_j)_+
=
\sum_{j=1}^{\ell}d_j(k-a_j)
=
k\dim W_\ell-c(W_\ell)
=
\Phi(\V,W_\ell,R).
\]
Therefore, $\dim K_W\le\max_{T\subseteq W}\Phi(\V,T,R)+(\lambda+1)\dim W$ whenever none of the $s$ exceptional events above occurs.

Since $s\le b$ and each exceptional event has probability at most $(2^b-1)\left(\frac{(4b)^{4b}mn}{\lambda(N-n+1)}\right)^{\lambda/(2b)}$, a union bound gives
\[
\Pr\!\left[
\dim K_W>
\max_{T\subseteq W}\Phi(\V,T,R)
+
(\lambda+1)\dim W
\right]
\le
b(2^b-1)
\left(
\frac{(4b)^{4b}mn}{\lambda(N-n+1)}
\right)^{\lambda/(2b)}.
\]
\end{proof}

We now prove the above-threshold direction. The idea is to compare the full space of matrices satisfying $\mathcal V$ with the subspaces in which two columns are equal. We show that, with high probability, every such collision subspace has strictly smaller dimension. Since there are fewer than $q$ possible pairs of columns, their union cannot cover the full space.

\begin{proof}[Proof of Theorem~\ref{thm:threshold-AG-codes}, part~2]
By the definition of $R_{\mathcal P}$, fix $\mathcal V\in\mathcal F$ and $U\in\mathcal L_{\mathrm{Dist}}(\mathbb F_q^b)$ such that $R_{\mathcal V,U}=R_{\mathcal P}$. Then, for every proper subspace $T\subsetneq U$, $\Phi(\mathcal V,U,R_{\mathcal P})\ge\Phi(\mathcal V,T,R_{\mathcal P}).$ Since $R\ge R_{\mathcal P}+\varepsilon$, it follows that
\begin{equation}
\label{eq:ag-high-rate-potential-gap}
\Phi(\mathcal V,U,R)-\Phi(\mathcal V,T,R)
\ge\varepsilon n(\dim U-\dim T).
\end{equation}

Matrices in $\mathcal C^b$ whose rows lie in $U$ form a space of dimension $k\dim U$. At coordinate $i$, requiring the row to lie in $V_i\cap U$ imposes at most $\dim U-\dim(V_i\cap U)$ linear conditions. Hence, deterministically,
\begin{equation}
\label{eq:ag-total-space-lower-bound}
\dim K_U\ge k\dim U-c(U)=\Phi(\mathcal V,U,R).
\end{equation}

For $1\le u<v\le b$, set $H_{u,v}:=\{x\in U:x_u=x_v\}.$ Because $U\in\mathcal L_{\mathrm{Dist}}(\mathbb F_q^b)$, each $H_{u,v}$ is a proper subspace of $U$. Applying~\eqref{eq:ag-high-rate-potential-gap} to $T\subseteq H_{u,v}$ gives
\begin{equation}
\label{eq:ag-collision-potential-gap}
\max_{T\subseteq H_{u,v}}\Phi(\mathcal V,T,R)
\le\Phi(\mathcal V,U,R)-\varepsilon n.
\end{equation}

Take $\lambda=\varepsilon n/(2b)$. Since $\varepsilon n\ge4b^2(b+1)$, we have $\lambda\ge2b(b+1)$ and $(\lambda+1)\dim H_{u,v}\le\frac{\varepsilon n}{2}+b<\varepsilon n.$ Apply Lemma~\ref{lem:ag-constrained-dimension} to each $H_{u,v}$. Together with~\eqref{eq:ag-collision-potential-gap} and~\eqref{eq:ag-total-space-lower-bound}, a union bound shows that, except with probability at most $\binom{b}{2}b(2^b-1)
\left(
\frac{(4b)^{4b}mn}{\lambda(N-n+1)}
\right)^{\lambda/(2b)}$,
\[
\dim K_{H_{u,v}}
\le\max_{T\subseteq H_{u,v}}\Phi(\mathcal V,T,R)
 +(\lambda+1)\dim H_{u,v}
<\Phi(\mathcal V,U,R)
\le\dim K_U
\]
for every $u<v$.

On this event, every $K_{H_{u,v}}$ is a proper subspace of $K_U$. Since $q>\binom{b}{2}$,
\[
\left|K_U\setminus\bigcup_{u<v}K_{H_{u,v}}\right|
\ge q^{\dim K_U}
 -\binom{b}{2}q^{\dim K_U-1}
>0.
\]
A matrix in this difference satisfies $\mathcal V$ and has pairwise distinct columns in $\mathcal C$. Thus $\mathcal C$ contains $\mathcal V$, and hence $\mathcal C\in\mathcal P$. Finally,
\[
\Pr[\mathcal C\notin\mathcal P]
\le\binom{b}{2}\,b(2^b-1)
\left(
\frac{2b(4b)^{4b}m}{\varepsilon(N-n+1)}
\right)^{\varepsilon n/(4b^2)},
\]
which proves part~2.
\end{proof}

\subsection{List-decoding and list-recovery consequences}
\label{subsec:AG-applications}

Combining the below threshold part of Theorem~\ref{thm:threshold-AG-codes} with the random linear codes results in Section~\ref{LCL-property}, we obtain the following list decoding and list recovery guarantees for random AG codes.

Let $F/\F_q$ be a function field, fix a rational place $P_\infty$ and let $N$ be the number of rational places distinct from $P_\infty$. For integers $0\le m<n\le N$, sample $P_1,\ldots,P_n$ uniformly without replacement, let $D:=P_1+\cdots+P_n$ and let $\cC:=\cC_{\mathcal{L}}(D;mP_\infty)\subseteq\F_q^n$ with $R:=\frac{\ell(mP_\infty)}{n}.$

In the statements below, $\delta$ is the gap appearing in the corresponding random linear codes result, while $\varepsilon$ is the additional rate loss in transferring the result to random AG codes. For an output list size $L$, set $b=L+1$. We assume
\begin{equation}
\label{eq:AG-decoding-size-assumptions}
q>\binom{b}{2},
\qquad
\varepsilon n\ge4b(b+1),
\qquad
N-n+1\ge
\frac{2(4b)^{4b}m}{\varepsilon}(2\ell)^{8b^2/\varepsilon},
\end{equation}
where $\ell=1$ for list decoding and average weight list decoding. Under these assumptions, the resulting failure probability is exponentially small in $n$.

For each list decoding or list recovery corollary below, let $\mathcal P$ denote the corresponding decoding-failure property. The RLC result used in each corollary shows that, at rate $R+\varepsilon$, $\Pr[\mathcal C_{\mathrm{RLC}}\in\mathcal P]=o(1).$ Hence, for all sufficiently large $n$, $R_{\mathcal P}>R+\frac{\varepsilon}{2}.$ Indeed, otherwise $R+\varepsilon\ge R_{\mathcal P}+\varepsilon/2$ and Theorem~\ref{thm:RLC-threshold} would imply $\Pr[\mathcal C_{\mathrm{RLC}}\in\mathcal P]=1-o(1),$ a contradiction.

Therefore, the random AG code of rate $R$ lies at least $\varepsilon/2$ below the threshold $R_{\mathcal P}$. For these decoding applications, the corresponding profile family has size at most $(2\ell)^{bn},$ with $\ell=1$ for list decoding; the same crude bound also applies to the average weight variants. Applying Theorem~\ref{thm:threshold-AG-codes} with gap $\varepsilon/2$ and using~\eqref{eq:AG-decoding-size-assumptions} gives $\Pr[\mathcal C\in\mathcal P]\le(2^b-1)(2\ell)^{-bn}\le 2^{-n}$ for all sufficiently large $n$.

\begin{cor}[List-decodability of random AG codes]
\label{cor:AG-list-decodability-sharp}
Fix $q$ and $\rho\in(0,1-1/q)$. Let $C_{\rho,q}$ be an integer at least the constant in Theorem~\ref{thm:RLC-list-decodability-constant-q}. For fixed $\delta,\varepsilon>0$ with $R=1-H_q(\rho)-\delta-\varepsilon>0$, set $L:=\left\lceil\frac{H_q(\rho)}{\delta}\right\rceil+C_{\rho,q}.$ Under~\eqref{eq:AG-decoding-size-assumptions}, for all sufficiently large $n$, $\Pr\bigl[\cC\text{ is }(\rho,L)\text{ list decodable}\bigr]\ge1-2^{-n}.$
\end{cor}

\begin{cor}[Average weight list-decodability of random AG codes]
\label{cor:AG-average-list-decodability}
Fix $q$ and $\rho\in(0,1-1/q)$, and let $C_{\rho,q}$ be the constant in Theorem~\ref{thm:RLC-average-list-decodability-constant-q}. For every sufficiently small fixed $\delta>0$ and fixed $\varepsilon>0$ with $R=1-H_q(\rho)-\delta-\varepsilon>0$, set $L:=\left\lceil\frac{C_{\rho,q}}{\delta}\right\rceil.$ Under \eqref{eq:AG-decoding-size-assumptions}, for all sufficiently large $n$, $\Pr\bigl[\cC\text{ is }(\rho,L)\text{ average weight list decodable}\bigr]\ge1-2^{-n}.$
\end{cor}

\begin{cor}[Average weight list decoding near the generalized Singleton bound]
\label{cor:AG-average-list-decodability-GSB}
Fix $R,\delta,\varepsilon>0$ with $R+\delta+\varepsilon<1$ and $L\in\N$. Assume
\[
q>\max\left\{
(32(L+1))^{2/\delta},
\frac{2(L+1)}{1-R-\delta-\varepsilon}
\right\},
\qquad
\rho:=\frac{L}{L+1}(1-R-\delta-\varepsilon).
\]
Under \eqref{eq:AG-decoding-size-assumptions}, for all sufficiently large $n$ with $n>4L/\delta$ and $(R+\varepsilon)n,\delta n\in\N$,
\[
\Pr\bigl[\cC\text{ is }(\rho,L)\text{ average weight list decodable}\bigr]\ge1-2^{-n}.
\]
\end{cor}
This follows from Theorem~\ref{thm:RLC-list-decodability-GSB} applied at rate $R+\varepsilon$ with gap $\delta$.

\begin{cor}[List-recoverability of random AG codes]
\label{cor:AG-list-recoverability-entropy}
Fix $q$, an integer $1\le\ell<q$, and $\rho\in(0,1-\ell/q)$. Let $C_{\rho,\ell,q}$ be the constant in Theorem~\ref{thm:RLC-list-recoverary}. For fixed $\delta,\varepsilon>0$ with $R=1-h_{q,\ell}(\rho)-\delta-\varepsilon>0$, set $L:=\left\lceil\frac{C_{\rho,\ell,q}}{\delta}\right\rceil.$ Under \eqref{eq:AG-decoding-size-assumptions}, for all sufficiently large $n$, $\Pr\bigl[\cC\text{ is }(\rho,\ell,L)\text{ list recoverable}\bigr]\ge1-2^{-n}.$
\end{cor}

\begin{cor}[List recovery with a smaller output list]
\label{cor:AG-list-recoverability-GG}
Fix $R,\delta,\varepsilon>0$ with $R+\delta+\varepsilon<1$ and $\ell\in\N$. Set $L:=\left\lceil
\left(\frac{\ell}{R+\delta+\varepsilon}\right)^%
{2(R+\delta+\varepsilon)/\delta}
\right\rceil$ and $\rho:=1-R-\delta-\varepsilon,$ and assume $q>(2L)^{8/\delta}$. Under \eqref{eq:AG-decoding-size-assumptions}, for all sufficiently large $n$ with $n\ge8(L^2+L)/\delta$ and $(R+\varepsilon)n\in\N$, $\Pr\bigl[\cC\text{ is }(\rho,\ell,L)\text{ list recoverable}\bigr]\ge1-2^{-n}.$
\end{cor}
This follows from Theorem~\ref{thm:RLC-list-recoverability-GG} applied at rate $R+\varepsilon$ with gap $\delta$.

\begin{cor}
\label{cor:AG-list-recoverability-LS}
Fix $R,\delta,\varepsilon>0$ with $R+\delta+\varepsilon<1$ and $\ell\in\N$. Set $L:=\left\lceil
\left(\frac{2\ell}{\delta}\right)^{2\ell/\delta}
\right\rceil,\, \rho:=1-R-\delta-\varepsilon$ and assume $q\ge\max\left\{
\ell^{8(R+\varepsilon)/\delta+6},
\ell\cdot2^{4/\delta}
\right\}.$ Under \eqref{eq:AG-decoding-size-assumptions}, for all sufficiently large $n$, $\Pr\bigl[\cC\text{ is }(\rho,\ell,L)\text{-list-recoverable}\bigr]
\ge1-2^{-n}.$
\end{cor}

This follows from Theorem~\ref{thm:RLC-list-recoverability-LS} applied at rate $R+\varepsilon$ with gap $\delta$.

\paragraph{Proximity gaps and correlated agreement.}
We first recall the curve versions of the two notions. 

\begin{defn}
Let $\mathcal C\subseteq\F_q^n$ be a linear code. For integers $t\ge0$ and $0\le E\le E^+\le n$, and $\tau\in[0,1]$, we say that $\mathcal C$ satisfies the $(t,E,E^+,\tau)$-\emph{curve proximity gap property} if the following holds.

For every curve of degree at most $t$ $\Gamma(\alpha)=\mathbf{u}_0+\alpha \mathbf{u}_1+\cdots+\alpha^t \mathbf{u}_t$ with $\alpha\in\F_q$ and $\mathbf u_0,\ldots,\mathbf u_t\in\F_q^n$, if $\bigl|\{\alpha\in\F_q:d(\Gamma(\alpha),\mathcal C)\le E\}\bigr|>\tau q,$ then $d(\Gamma(\alpha),\mathcal C)\le E^+$ for every $\alpha\in\F_q.$
\end{defn}

\begin{defn}[Definition~5.1~\cite{yuan2026proximity}]
Let $\mathcal C\subseteq\F_q^n$ be a linear code, with $t,E,E^+,\tau$ as above. We say that $\mathcal C$ satisfies the $(t,E,E^+,\tau)$-\emph{curve correlated agreement property} if, for every curve $\Gamma(\alpha)=\mathbf{u}_0+\alpha \mathbf{u}_1+\cdots+\alpha^t \mathbf{u}_t$ with $\mathbf u_0,\ldots,\mathbf u_t\in\F_q^n$, whenever $\bigl|\{\alpha\in\F_q:d(\Gamma(\alpha),\mathcal C)\le E\}\bigr|>\tau q,$ there exist codewords $\mathbf{c}_0,\ldots,\mathbf{c}_t\in\mathcal C$ such that $\left|\bigcup_{j=0}^t\mathrm{supp}(\mathbf u_j-\mathbf c_j)\right|\le E^+.$
\end{defn}

Yuan and Zhu~\cite{yuan2026proximity} prove proximity gap and correlated agreement results for random linear codes and random Reed--Solomon codes. In their random Reed--Solomon argument, a failure of curve correlated agreement is reduced to the occurrence of a suitable local profile. Replacing the random RS profile bound in that argument with the random AG profile bound proved above gives the following analogous result for random one-point AG codes.

\begin{cor}[Curve correlated agreement for random AG codes]
\label{cor:AG-curve-CA}
Fix $R\in(0,1)$, $0<\varepsilon<(1-R)/2$, and an integer $t\ge1$. Set $\beta:=1+\left\lceil\frac{4(t+1)(1-R)}{\varepsilon}\right\rceil .$ There exists a constant $C=C_{R,t}>0$ such that the following holds for all sufficiently large $n$.

Assume $q>\binom{\beta}{2},\,\log q=o(n)$ and $N-n+1\ge m\,2^{C/\varepsilon^3}.$ Let $0<\rho\le1-R-\varepsilon,\,E:=\lfloor\rho n\rfloor$ and $d_0:=\left\lfloor\left(1-R-\frac{\varepsilon}{4}\right)n\right\rfloor.$ For every integer $E\le E^+<d_0$, define $B_{E,E^+}:=\frac{E^++1}{E^+-E+1}$ and
\[
\tau_t
:=
\min\left\{
1,\,
\frac{tB_{E,E^+}}{q}
\left(
\frac{d_0}{d_0-E}
\right)^{
\left\lceil
4(t+1)(1-R)/\varepsilon
\right\rceil}
\right\}.
\]

Then, with probability at least $1-\exp\bigl(-\Omega_{R,\varepsilon,t}(n)\bigr),$ the random one-point AG code $\mathcal C$ satisfies the $(t,E,E^+,\tau_t)$-curve correlated agreement property. Consequently, it also satisfies the $(t,E,E^+,\tau_t)$-curve proximity gap property.
\end{cor}

When $E^+=E$, the soundness is $O(n/q)$. If $E^+-E\ge \eta n$ for some fixed $\eta>0$, then the soundness improves to $O(1/q)$. The hidden constants depend on $R,\varepsilon,t$, and in the second case also on $\eta$.

\bibliographystyle{alpha}
\bibliography{refs}


\end{document}